\documentclass[10pt, conference, letterpaper]{IEEEtran}
\usepackage{amsmath}
\usepackage{graphicx}
\usepackage{subfigure}
\usepackage{color}
\usepackage{algorithm}
\usepackage[noend]{algorithmic}
\usepackage{float}
\usepackage{authblk}
\usepackage{amssymb}
\usepackage{enumerate}
\usepackage{url}
\usepackage[T1]{fontenc}
\usepackage{amsthm}

\newtheorem{theorem}{Theorem}

\newtheorem{definition}{Definition}

\begin{document}

\title{Multicast Traffic Engineering for Software-Defined Networks%
}

\author[*]{Liang-Hao Huang}
\author[*]{Hsiang-Chun Hsu}
\author[*]{Shan-Hsiang Shen}
\author[*]{De-Nian Yang}
\author[*$\dagger$]{Wen-Tsuen Chen}
\affil[*]{Institute of Information Science, Academia Sinica, Taipei, Taiwan}
\affil[$\dagger$]{Department of Computer Science, National Tsing Hua University, Hsinchu, Taiwan}
\affil[ ]{}
\affil[ ]{Email: lhhuang@iis.sinica.edu.tw, hchsu0222@gmail.com, \{sshen3, dnyang, chenwt\}@iis.sinica.edu.tw}

\renewcommand \Authands{ and }
%\author{}
\maketitle

\begin{abstract}
Although Software-Defined Networking (SDN) enables flexible network resource
allocations for traffic engineering, current literature mostly focuses on
unicast communications. Compared to traffic engineering for multiple unicast
flows, multicast traffic engineering for multiple trees is very challenging
not only because minimizing the bandwidth consumption of a single multicast
tree by solving the Steiner tree problem is already NP-Hard, but the Steiner
tree problem does not consider the link capacity constraint for multicast
flows and node capacity constraint to store the forwarding entries in Group
Table of OpenFlow. In this paper, therefore, we first study the hardness
results of scalable multicast traffic engineering in SDN. We prove that
scalable multicast traffic engineering with only the node capacity
constraint is NP-Hard and not approximable within $\delta $, which is the
number of destinations in the largest multicast group. We then prove that
scalable multicast traffic engineering with both the node and link capacity
constraints is NP-Hard and not approximable within any ratio. To solve the
problem, we design a $\delta $-approximation algorithm, named Multi-Tree
Routing and State Assignment Algorithm (MTRSA), for the first case and
extend it to the general multicast traffic engineering problem. The simulation and implementation
results demonstrate that the solutions obtained by the proposed algorithm
outperform the shortest-path trees and
Steiner trees. Most importantly, MTRSA is computation-efficient and can be
deployed in SDN since it can generate the solution with numerous trees in a
short time.
\end{abstract}

%\author{\IEEEauthorblockN{Shan-Hsiang Shen$^\ast$,~Liang-Hao Huang$\ast$,~De-Nian Yang$\ast$, and ~Wen-Tsuen Chen}
%\IEEEauthorblockA{Academia Sinica, Taipei, Taiwan\\
%\{sshen3, lhhuang, dnyang, chenwt\}@iis.sinica.edu.tw}}

%\author[*]{Shan-Hsiang Shen}
%\author[*]{Liang-Hao Huang}
%\author[*]{De-Nian Yang}
%\author[*$\dagger$]{Wen-Tsuen Chen}
%\affil[*]{Institute of Information Science, Academia Sinica, Taipei, Taiwan}
%\affil[$\dagger$]{Department of Computer Science, National Tsing Hua University, Hsinchu, Taiwan}
%\affil[ ]{}
%\affil[ ]{Email: \{sshen3, lhhuang, dnyang, chenwt\}@iis.sinica.edu.tw}

\renewcommand \Authands{ and }

%\author{}


%\author{\IEEEauthorblockN{Liang-Hao Huang,~Hui-Ju Hung,~Chih-Chung Lin, and~De-Nian Yang}
%\IEEEauthorblockA{Academia Sinica, Taipei, Taiwan\\
%\{lhhuang, hjhung, chchlin, dnyang\}@iis.sinica.edu.tw}}

%\author{\IEEEauthorblockN{DLiang-Hao Huang, Hui-Ju Hung, Chih-Chung Lin, De-Nian Yang}
%\IEEEauthorblockA{
%Email: dnyang@iis.sinica.edu.tw} }


\section{Introduction}

\IEEEPARstart{S}{oftware-defined} networking (SDN) provides a new
centralized architecture with flexible network resource management to
support a huge amount of data transmission~\cite{mckeown_openflow:_2008}.
Different from legacy networks, SDN separates the control plane from
switches and allows the control plane to be programmable to efficiently
optimize the network resources. OpenFlow~\cite{mckeown_openflow:_2008} in
SDN includes two major components: controllers (SDN-Cs) and forwarding
elements (SDN-FEs). Controllers are in charge of handling the control plane
and install forwarding rules based on different policies, while forwarding
elements in switches deliver packets according to the rules specified by the
controllers. Compared with the current Internet, routing paths no longer
need to be the shortest ones, and the paths can be distributed more flexibly
inside the network. It has been demonstrated that SDN provides a better
overview of network topologies and enables centralized computation for
traffic engineering for multiple unicast flows~\cite%
{agarwal_traffic_2013,Lazaris:2014:TSS:2674005.2675011,Nguyen:2014:ORP:2620728.2620753}%
. However, multicast traffic engineering for multiple multicast trees in SDN
has attracted much less attention in previous studies.

Compared to unicast, multicast has been shown in empirical studies to be
able to effectively reduce overall bandwidth consumption in backbone
networks by around 50\% \cite{malli_benefit_1998}. It employs a multicast
tree, instead of disjoint unicast paths, from the source to all destinations
of a multicast group, in order to avoid unnecessary traffic duplication. The
current Internet multicast standard, i.e., PIM-SM \cite{Imai2007}, employs
a shortest-path tree to connect the source and destinations, and traffic
engineering is difficult for PIM-SM since the path from the source to each
destination is the shortest one. A shortest-path tree tends to lose many good
opportunities to reduce the bandwidth consumption by sharing more common
edges among the paths to different destinations. In contrast, to minimize
the bandwidth consumption, a Steiner tree (ST)~\cite{hwang_steiner_1992} in
Graph Theory minimizes the number of edges in a multicast tree.
Nevertheless, ST only focuses on the routing of a multicast tree, instead of
jointly optimizing the resource allocations of all trees. Therefore, when
the network is heavily loaded, a link will not be able to support a large
number of STs that choose the link. Most importantly, Group Table of an
SDN-FE will be insufficient to store the forwarding entries of the STs due
to the small TCAM size \cite{Yu:2010:SFN:1851182.1851224}.

Compared to the shortest-path routing in unicast, unicast traffic
engineering in SDN is more difficult to aggregate multiple flows in Flow
Table of an SDN-FE, and the scalability has been regarded as a serious issue
in the deployment of SDN for a large network \cite{agarwal_traffic_2013, kanizo_palette:_2013}.
The scalability problem for multicast communications is even more serious
since the number of possible multicast groups is $O(2^{n})$, where $n$ is the
number of nodes in a network, and the number of possible unicast connections
is $O(n^{2})$. To remedy this issue, a promising way is to exploit the 
\textit{branch forwarding technique} \cite{Yang2008, YangLiao2008,
Stoica2000}, which stores the multicast forwarding entries in only
the \textit{branch nodes}, instead of every node, of a multicast tree, where
a branch node in a tree is the node with at least three incident edges. The
branch forwarding technique can remedy the multicast scalability problem
since packets are forwarded in a unicast tunnel from the logic port of a
branch node in SDN-FE \cite{mckeown_openflow:_2008} to another branch node. In other
words, all nodes in the path exploit unicast forwarding in the tunnel and
are no longer necessary to maintain a forwarding entry for the multicast
group. Furthermore, when a branch node is not multicast capable for a tree
(ex. Group Table is full in this paper), local \textit{unicast tunneling}
from a nearby multicast capable node has been proposed in MBONE \cite%
{Almeroth2000} and PIM-SM\footnote{%
http://www.cisco.com/c/en/us/td/docs/ios-xml/ios/ipmulti%
\_pim/configuration/xe-3s/imc-pim-xe-3s-book/imc\_tunnel.html} to allow multiple unicast tunnels to
pass through the branch node to other nodes in the tree (an example will be presented later in this section). Nevertheless, compared to multicast,
it is envisaged that local unicast tunneling will incur more bandwidth
consumption since duplicated packets will be delivered in a link. Therefore,
there is a trade-off between the link capacity and node capacity, because
each branch node can act as either a \textit{branch state node} with the
corresponding multicast forwarding entry stored in Group Table or a \textit{branch
stateless node} that exploits the unicast tunneling strategy. 

In comparison with the ST problem, \textit{scalable multicast traffic
engineering}, which jointly allocates the network resources for multiple
trees, is much more challenging because both the \textit{link capacity} and 
\textit{node capacity} constraints are involved in the problem. The link
capacity constraint states that the total rate of all multicast trees on
each link should not exceed the corresponding link capacity, while the node
capacity constraint ensures that Group Table of each node is sufficiently
large to support the multicast trees with the node as a branch state node.
Moreover, scalable multicast traffic engineering\textit{\ }with branching
forwarding and unicast tunneling techniques is able to allocate the network
resources more flexibly. When Group Table of a node is full, unicast
tunneling moves the resource requirement from the node to its incident links,
whereas the rerouting of the tree is also promising by exploiting the
resources of the nearby nodes and links. Therefore, it is necessary for
scalable multicast traffic engineering to carefully examine both the routing
and the allocation of the branch state nodes of all multicast trees. In this
paper, we explore the Scalable Multicast Traffic Engineering (SMTE) problem
for SDNs. Given the data rate requirement of each multicast tree, SMTE aims
to minimize the total bandwidth cost of all trees, by finding a tree
connecting the source and destinations of each group and assigning the
branch state nodes for each tree, such that both the link capacity and node
capacity constraints can be ensured. 

\bigskip Fig. \ref{fig1:subfig} presents an illustrative example. Fig. %
\ref{fig1:subfig:a} is the original network with the unit bandwidth cost
specified beside each link. The bandwidth cost of each link is the total
bandwidth consumption of the link multiplied by the unit bandwidth cost. The
node capacity of each node is $1$. The link capacity of edge $e_{s,a}$ is $1$%
, and the link capacities of the other edges are $\infty $ in this example.
There are two multicast trees with both flow rates as $1$. The source of the
first tree is $s_{1}=s$, and its destination set is $D_{1}=\{d_{1},d_{2},%
\dots ,d_{7}\}$. The source of the second tree is $s_{2}=s$, and its
destination set is $D_{2}=\{d_{1}^{\prime },d_{2}^{\prime },\dots
,d_{7}^{\prime }\}$. Fig. \ref{fig1:subfig:b} shows the first shortest-path
tree (blue) and the second shortest-path tree (red). The branch nodes and
branch state nodes of the first tree are $\left \{ c,u,v\right \} $ and $%
\left \{ u\right \} $, respectively. The branch nodes and branch state nodes
of the second tree are $\left \{ c,v\right \} $ and $\left \{ c,v\right \} $,
respectively. Note that $v$ is not assigned as a branch state node of the
first tree, and $u$ thus needs to exploit unicast tunneling to $d_{6}$ and $%
d_{7}$ directly. Therefore, traffic of the first tree are duplicated in edge 
$e_{u,v}$. On the other hand, if $v$ was assigned as a branch state node for
the first tree, traffic duplication in $e_{u,v}$ would be more serious for
the second tree since $v$ has three downstream nodes $d_{5}^{\prime }$, $%
d_{6}^{\prime }$, $d_{7}^{\prime }$. The total bandwidth cost of the two
shortest-path trees in Fig. \ref{fig1:subfig:b} is $99$. 

Afterward, Fig. \ref{fig1:subfig:c} shows the first Steiner tree (blue) and
the second Steiner tree (red), and the branch state nodes of the two trees
are also $\left \{ u\right \} $ and $\left \{ c,v\right \} $, respectively.
The total bandwidth cost of the trees in Fig. \ref{fig1:subfig:c} is $103$. 
Note that the total bandwidth cost in Steiner trees is higher since the 
assignment of branch state nodes are not carefully examined.
Finally, Fig. \ref{fig1:subfig:d} presents the first tree (blue) and the
second tree (red) in SMTE with the same branch state nodes specified above.
The total bandwidth cost of the trees in Fig. \ref{fig1:subfig:c} is $79$, and here $u$ 
is directed connected to $d_{1}$, $d_{2}$, and $d_{3}$ to avoid unicast 
tunneling, even though the edge cost is higher (i.e., 2) compared to the 
cost (i.e., 1) of the edge from $v$ in the first Steiner tree.
Therefore, this example manifests that it is necessary to consider the tree 
routing and the assignment of branch state nodes of multiple trees jointly 
for scalable multicast traffic engineering.

\begin{figure}[t]
\centering
\subfigure[Original network]{
    \label{fig1:subfig:a}
    \includegraphics[height = 3cm, width = 4cm]{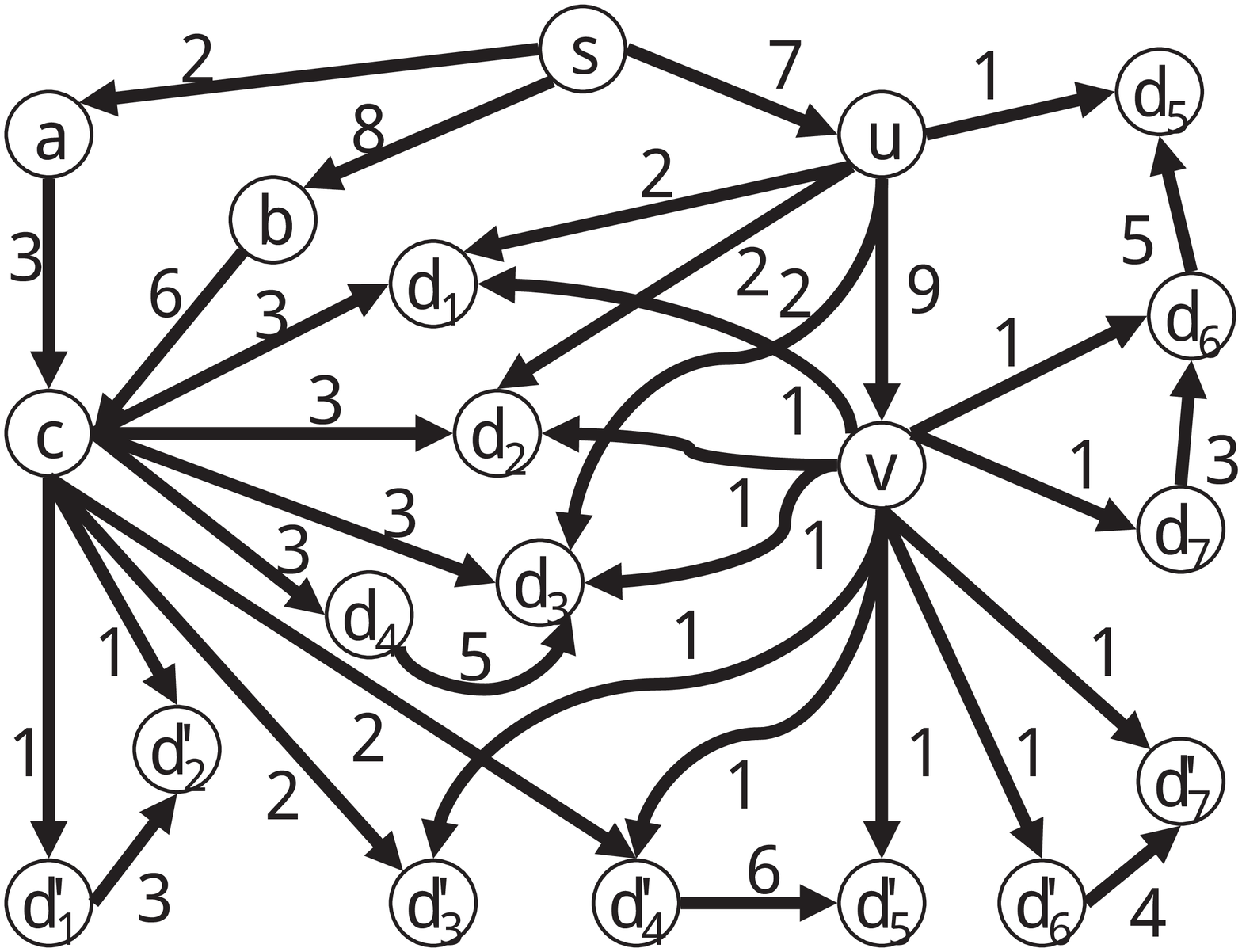}} \smallskip 
    %\vspace{-0.25in}
\subfigure[Shortest-path trees]{
    \label{fig1:subfig:b}
    \includegraphics[height = 3cm, width = 4cm]{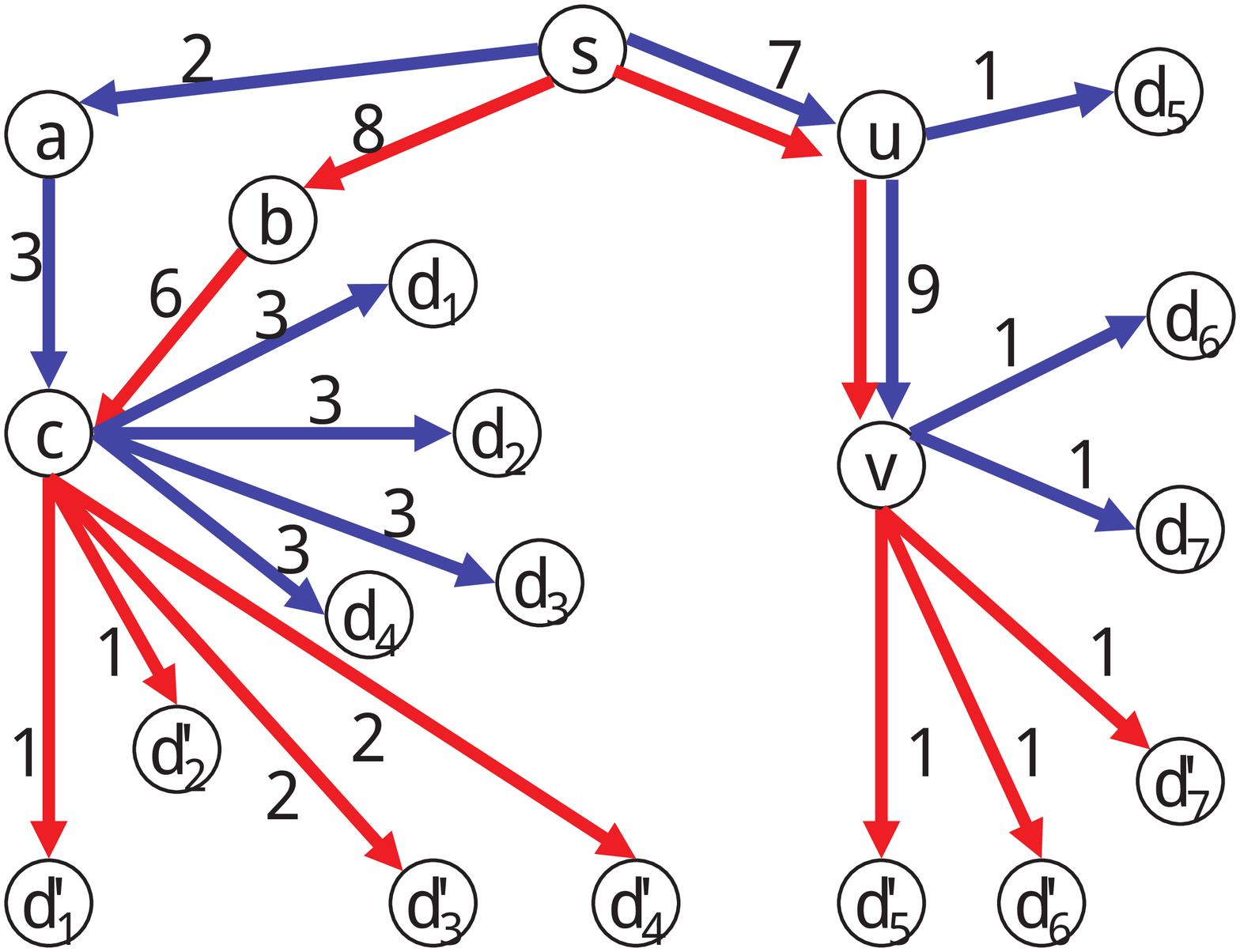}} \smallskip 
    %\vspace{-0.25in}
\subfigure[Steiner trees]{
    \label{fig1:subfig:c}
    \includegraphics[height = 3cm, width = 4cm]{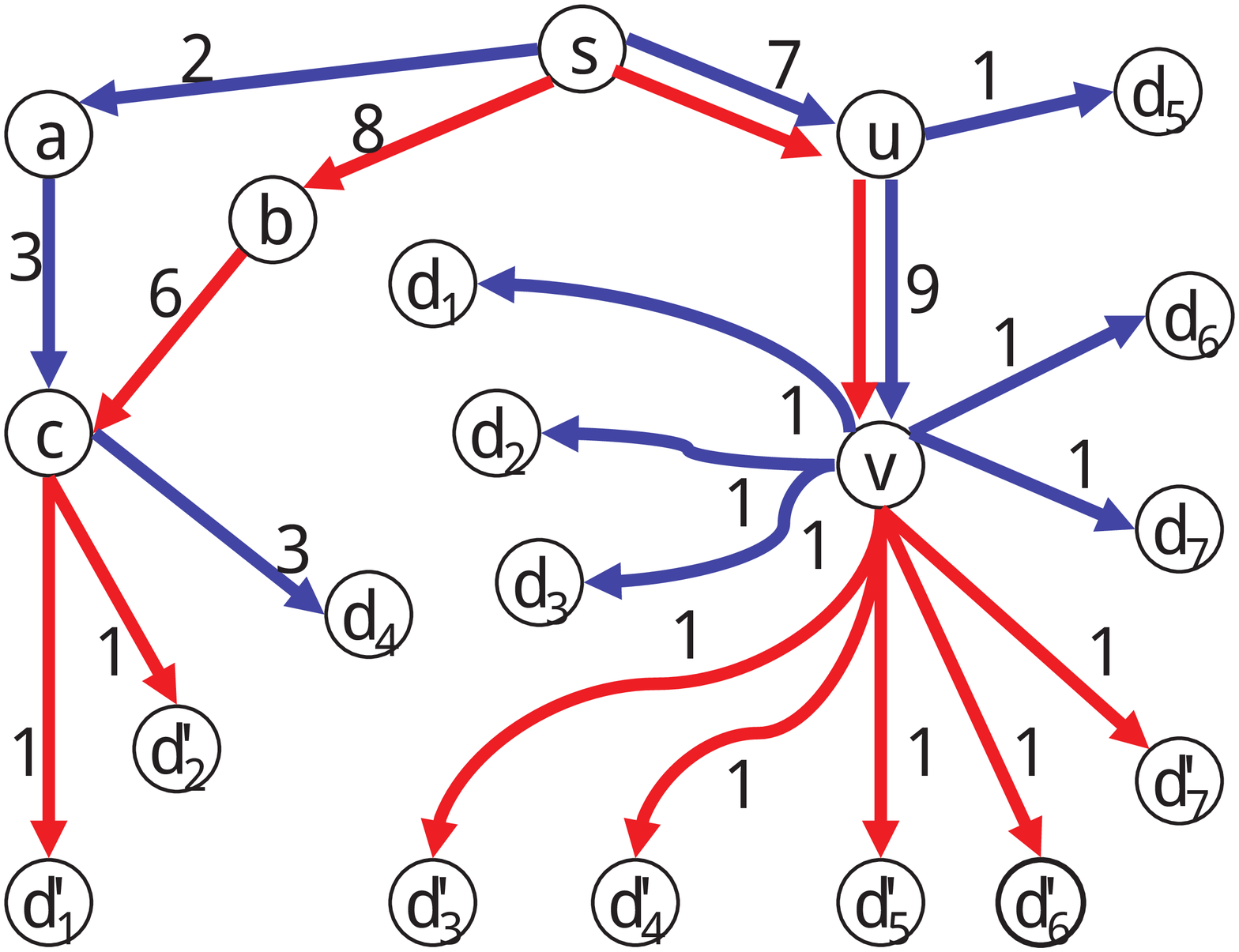}} \smallskip 
\subfigure[MTRSA]{
    \label{fig1:subfig:d}
    \includegraphics[height = 3cm, width = 4cm]{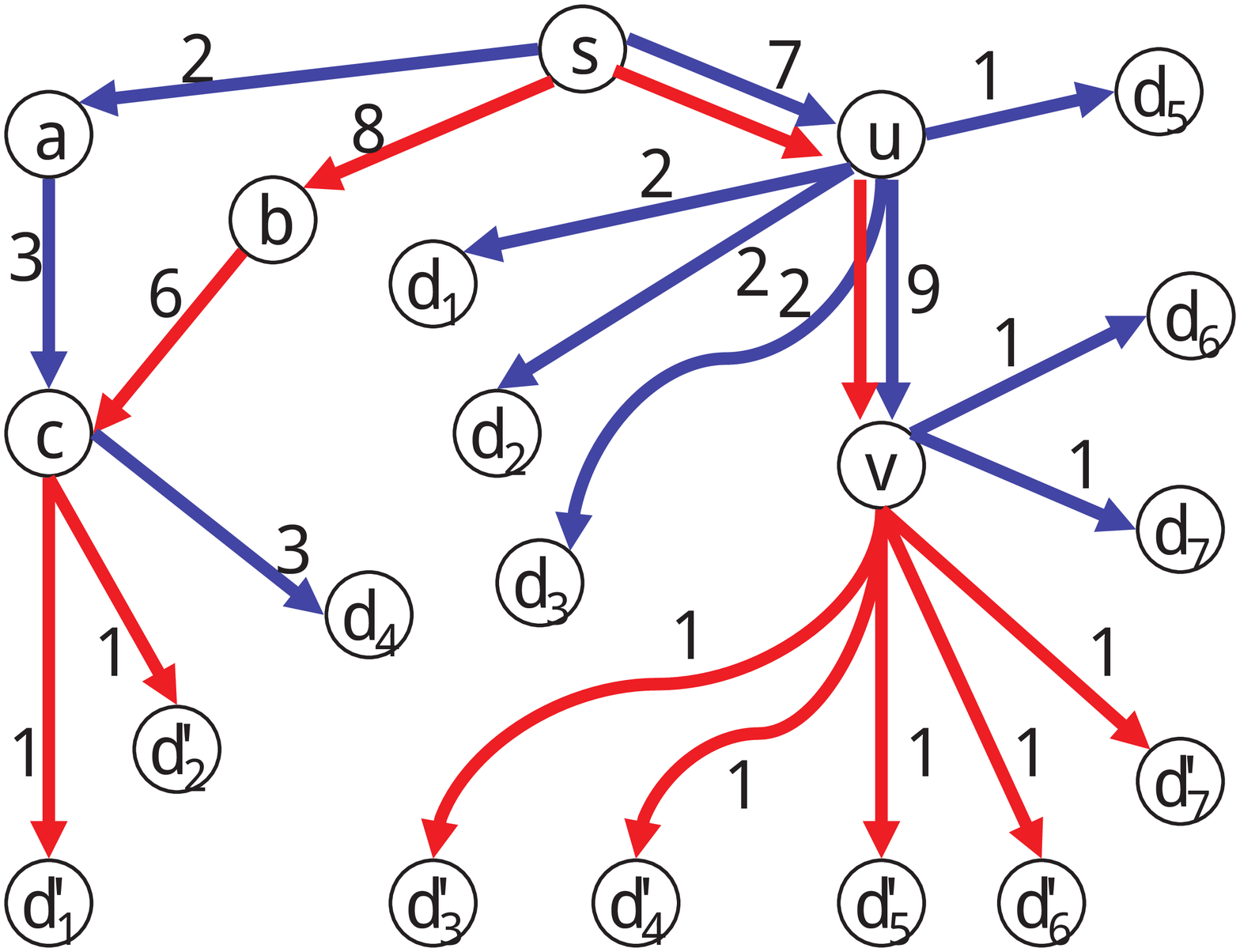}} \smallskip \vspace{%
-0.1in} %\vspace{-0.25in}
\caption{Comparison of different strategies for multicast traffic engineering}
\label{fig1:subfig}
\end{figure}

SMTE is very challenging. The ST problem is NP-Hard but can be approximated
within the ratio 1.55 \cite{robins_improved_2000} and is thus in APX of Complexity
Theory. In other words, there exists an approximation algorithm for ST that
can find a tree with the total cost at most 1.55 times of the optimal
solution. In contrast, we first prove that SMTE-N (i.e., SMTE with only the
node capacity constraint, while the link capacity constraint is relaxed) is
NP-Hard but cannot be approximated within $\delta $, which denotes the
number of destinations of the largest multicast group. Afterward, we prove
the SMTE (i.e., with both the link and node capacity) cannot be approximated
within any ratio. To solve SMTE-N, we propose a $\delta $-approximation
algorithm, named \textit{Multi-Tree Routing and State Assignment Algorithm
(MTRSA)}, which can be deployed in SDN-C. MTRSA includes two phases:
Multi-Tree Routing Phase and State-Node Assignment Phase, to effectively
minimize the total bandwidth cost of all trees according to the node
capacity constraint. We first focus on the node capacity (i.e., SMTE-N),
instead of the link capacity, because the scalability in Group Table is
unique and crucial for SDN and has not been explored in previous studies of multicast
tree routing for other networks. Since no $(\delta ^{1-\epsilon })$%
-approximation algorithm exists in SMTE-N for arbitrarily small $\epsilon >0$%
, MTRSA achieves the best approximation ratio. Afterward, we extend MTRSA to
support SMTE with the link capacity constraint.

The rest of this paper is organized as follows. Section 2 introduces the
related work. Section 3 and 4 formulate SMTE with Integer Programming and
describe the hardness results. We present the algorithm design of MTRSA in
Section 5, and Section 6 shows the simulation and implementation results on real topologies.
Finally, Section 7 concludes this paper.

\section{Related Work\label{sec: RelatedWork}}

\label{sec:related} The issues of traffic engineering for \textit{unicast
traffic} in SDN have attracted a wide spectrum of attention in the
literature. Sushant et al. \cite{jain_b4:_2013} developed private WAN of
Google Inc. with the SDN architecture. Qazi et al. \cite{qazi_simple-fying_2013} designed a
new system in SDN to control the middleboxes, and Mckeown et al. \cite%
{mckeown_openflow:_2008} studied the performance of OpenFlow in
heterogeneous SDN switches. Agarwal et al. \cite{agarwal_traffic_2013}
presented unicast traffic engineering in an SDN network with only a few
SDN-FEs, while the other routers in the network followed a standard routing
protocol, such as OSPF. However, the above studies focused on only unicast
traffic engineering, and multicast traffic engineering for multiple
multicast trees in SDN has attracted much less attention.

To support the multicast communications, the current multicast routing
standard PIM-SM~\cite{Imai2007} relies on unicast routing
protocols to discover the shortest paths from the
source to the destinations for building a shortest-path tree (SPT). However,
SPT is not designed to support traffic engineering. Although the Steiner
tree (ST) \cite{hwang_steiner_1992} minimizes the tree cost and the volume
of traffic in a network, ST is computationally intensive and is not
adopted in the current Internet standard. Overlay ST \cite%
{yang_bandwidth-efficient_2007, aharoni_restricted_1998}, on the other hand,
presents an alternative way to construct a bandwidth-efficient multicast
tree in the P2P environment. However, the path between any two P2P clients
is still a shortest path in Internet, and it is, therefore, difficult to
optimize the routing of the P2P tree. Most importantly, both SPT and ST are
designed to find the routing of a tree, instead of jointly optimizing the
resource allocation of multiple trees.

Flow table scalability is crucial to support large-scale SDN networks due to
the limited TCAM size. Kanizo et al.~\cite{kanizo_palette:_2013}, who 
showed that the major bottleneck in SDN is the restricted table sizes,
proposed a framework called Patette to decompose a large SDN table and
distribute its entries across a network. 
%To address the issues, many studies focus on traffic engineering for SDN networks.
Leng et al.~\cite{sdnflowtable} proposed a flow table reduction scheme (FTRS)
to reduce flow table usage with omnipotent controller functions. DIFANE~\cite%
{Yu:2010:SFN:1851182.1851224} distributed the flow entries to multiple SDN
switches. Zhang at al.~\cite{nfvmulticast} built a multicast topology (single backbone tree) for  NFV, while Craig et al. adjusted the link weights for shortest-path trees in SDN~\cite{loadbalance15}. Huang et al. also tried to optimize the routing of single multicast tree in SDN~\cite{7037084,7218381}.
Nevertheless, the above studies were not
designed for minimizing the total resource consumption in multicast traffic 
engineering with multiple trees subject to both the node and link capacity constraints in SDN.

\section{Problem Formulation\label{sec: ProblemFormulation}}

In this paper, we explore the \textit{Scalable Multicast Traffic Engineering}
(SMTE) problem for SDN. Given the data rate requirement of each multicast
group, SMTE aims to minimize the total bandwidth consumption of all
multicast groups in the network, by finding a tree connecting the source and
destinations of each group, and assigning the branch state nodes for each
tree, such that the number of multicast forwarding states will not
exceed the size of Group Table in each node, and the total multicast flows
on each edge will not exceed the link capacity. Note that a branch node can
only facilitate unicast tunneling for a multicast group if it is not
assigned as a branch state node in the corresponding multicast tree.

More specifically, given a network $G(V,E)$, where $V$ and $E$ denote the
set of nodes and directed edges, respectively, let $b_{v}$ denote the
maximal number of branch state nodes that can be maintained by node $v%
\footnote{%
In the following, we first assume that the memory size allocated in Group
Table\ to maintain the branch state node of each multicast tree is the same,
and later we extend it to the general scenario that supports different
memory sizes for different multicast trees according to the degrees of the node in the trees \cite{mckeown_openflow:_2008} in Section \ref{sec:
extension}.}$. Let $N_{v}^{+}$ ($N_{v}^{-}$) denote the set of out-neighbor
(in-neighbor) nodes of $v$ in $G$. Node $u$ is in $N_{v}^{+}$ ($N_{v}^{-}$)
if $e_{v,u}$ ($e_{u,v}$) is a directed edge from $v$ to $u$ (from $u$ to $v$%
) in $E$, and $c_{u,v}$ is the capacity of $e_{u,v}$, while $k_{u,v}$ is the
unit bandwidth cost of $e_{u,v}$. Let $\mathcal{T}=(T_{1},T_{2},\dots
,T_{t}) $ denote the set of multicast trees, while $s_{i}$ acts as the root
of tree $T_{i}\in \mathcal{T}$, i.e., the source with data rate $f_{i}$, and
the destination set $D_{i}$ contains the set of destinations in $T_{i}\in 
\mathcal{T}$. In the following, we first formally define SMTE, while the
derivation of the bandwidth consumption will be explained later in this
section in the proposed Integer Programming formulation. Dynamic group 
membership with user join and leave will be discussed later in Section~\ref{sec: extension}.

\begin{definition}
For network $G(V,E)$ and multicast groups $\mathcal{T}$, SMTE is to find the
routing of each tree $T_i$ in $\mathcal{T}$ spanning $s_{i}$ and $D_{i}$ and
assign the branch state nodes in $T_i$ to minimize the total bandwidth cost, such that each node $u$ acts as the
branch state nodes of at most $b_{u}$ trees, and total multicast bandwidth
consumption in each edge $e_{u,v}$ is at most $c_{u,v}$.
\end{definition}

In the following, we present the Integer Programming (IP) formulation for
SMTE. SMTE includes the following binary decision variables to find the
routing of each multicast tree and the assignment of branch state nodes. Let
binary variable $\pi _{i,d,u,v}$ denote if edge $e_{u,v}$ is in the path
from $s_{i}$ to a destination node $d$ in $D_{i}$ in $T_{i}$. Let integer
variable $\varepsilon _{i,u,v}$ denote the number of times that each packet
of $T_{i}$ is sent in edge $e_{u,v}$ via multicast (once) or unicast
tunneling (multiple times according to the number of tunnels). Let binary variable $\beta _{i,v},$ denote if $v$
is a branch state node in $T_{i}$. Intuitively, when we are able to find the
path from $s_{i}$ to each destination node $d$ of $T_{i}$ with $\pi
_{i,d,u,v}=1$ on every edge $e_{u,v}$ in the path, together with the set of state
branch nodes $\beta _{i,v}$, the routing of the tree (the set of edges $%
e_{u,v}$ with $\varepsilon _{i,u,v}\geq 1$) can be constructed according to
the paths from $s_{i}$ to all destination nodes in $D_{i}$.

The objective function of the IP formulation for SMTE is as follows.%
\begin{equation*}
\min \sum \limits_{1\leq i\leq t}\sum \limits_{e_{u,v}\in E}f_{i}\times
k_{u,v}\times \varepsilon _{i,u,v}.
\end{equation*}%
The objective function minimizes the total bandwidth cost of all multicast
trees. For each tree $T_{i}$, the following constraints first describe the
routing assignment (i.e., $\pi _{i,d,u,v}$) for the path connecting the
source $s_{i}$ and each destination in $D_{i}$. Afterwards, we assign the
branch nodes (i.e., $\beta _{i,u}$) in different nodes and then derive the
bandwidth consumption (i.e., $\varepsilon _{i,u,v}$) of $T_{i}$ via
multicast and unicast tunneling.

\begin{center}
\begin{tabular}{cc}
$\sum \limits_{v\in N_{s_{i}}^{+}}\pi _{i,d,s_{i},v}-\sum \limits_{v\in
N_{s_{i}}^{-}}\pi _{i,d,v,s_{i}}=1$, $\forall 1\leq i\leq t,d\in D_{i},$ & (1) \\ 
\vspace{2pt} $\sum \limits_{u\in N_{d}^{-}}\pi _{i,d,u,d}-\sum \limits_{u\in
N_{d}^{+}}\pi _{i,d,d,u}=1$, $\forall 1\leq i\leq t,d\in D_{i},$ & (2) \\ 
\vspace{2pt} $\sum \limits_{v\in N_{u}^{-}}\pi _{i,d,v,u}=\sum \limits_{v\in
N_{u}^{+}}\pi _{i,d,u,v}$, &  \\ 
$\forall 1\leq i\leq t,d\in D_{i},u\in V,u\neq d,u\neq s_{i},$ & (3) \\ 
\vspace{2pt} $\pi _{i,d,u,v}\leq \varepsilon _{i,u,v}$, $\forall 1\leq i\leq
t,d\in D_{i},\forall e_{u,v}\in E,$ & (4) \\ 
\vspace{2pt} $-|D_{i}|^2\times \beta _{i,u}+\sum \limits_{v\in
N_{u}^{+}}\varepsilon _{i,u,v}\leq \sum \limits_{v\in N_{u}^{-}}\varepsilon
_{i,v,u}$, & (5) \\ 
$\forall 1\leq i\leq t,u\in V,u\neq s_{i},$ &  \\ 
$\sum \limits_{1\leq i\leq t}\beta _{i,u}\leq b_{u}$, $\forall u\in V,$ & (6)
\\ 
$\sum \limits_{1\leq i\leq t}f_{i}\times \varepsilon _{i,u,v}\leq c_{u,v}$, $%
\forall e_{u,v}\in E.$ & (7)%
\end{tabular}
\end{center}

The first three constraints, i.e., (1), (2), and (3), are the
flow-continuity constraints for each tree $T_i$ to find the path from $s_{i}$
to every destination node $d$ in $D_{i}$. More specifically, $s_{i}$ is the
source node, and constraint (1) states that the net outgoing flow from $%
s_{i} $ is one, implying that at least one edge $e_{i,s_i,v}$ from $s_{i}$
to any neighbor node $v$ needs to be selected with $\pi _{i,d,s_i,v}=1$.
Note that here decision variables $\pi _{i,d,s_i,v}$ and $\pi _{i,d,v,s_i}$
are two different variables because the flow is directed. On the other hand,
every destination node $d$ is the flow destination, and constraint (2)
ensures that the net incoming flow to $d$ is one, implying that at least one
edge $e_{i,u,d}$ from any neighbor node $u$ to $d$ must be selected with $%
\pi _{i,d,u,d}=1$. For every other node $u$, constraint (3) guarantees that $%
u$ is either located in the path or not. If $u$ is located in the path, both
the incoming flow and outgoing flow for $u$ are at least one, indicating
that at least one binary variable $\pi _{i,d,v,u}$ is $1$ for the incoming
flow, and at least one binary variable $\pi _{i,d,u,v}$ is $1$ for the
outgoing flow. Otherwise, both $\pi _{i,d,v,u}$ and $\pi _{i,d,u,v}$ are $0$%
. Note that the objective function will ensure that $\pi _{i,d,v,u}=1$ for
at most one neighbor node $v$ to achieve the minimum bandwidth consumption.
In other words, both the incoming flow and outgoing flow among $u$ and $v$
cannot exceed $1$.

Constraints (4) and (5) are formulated to find the routing of the tree and
its corresponding branch state nodes, i.e., $\varepsilon _{i,u,v}$ and $%
\beta _{i,u}$. Constraint (4) states that $\varepsilon _{i,u,v}$ is at least 
$1$ if edge $e_{u,v}$ is included in the path from $s_{i}$ to at least one $%
d $, i.e., $\pi _{i,d,u,v}=1$. The tree $T_{i}$ is the union of the paths
from $s_{i}$ to all destination nodes in $D_{i}$. Constraint (5) is the most
crucial one. For each node $u$ in $T_{i}$, if it is not a branch state node,
i.e., $\beta _{i,u}=0$, $u$ does not maintain a forwarding entry of $T_{i}$
in Group Table and thereby facilitates unicast tunneling. In this case,
constraint (5) and the objective function guarantee that the number of
packets received from an incoming link $e_{v,u}$ must be the summation of
the number of packets sent to every outgoing link $e_{u,v}$. By contrast,
when $\beta _{i,u}=1$, constraint (5) becomes redundant because the
Left-Hand-Side (LHS) is smaller than $0$ and thereby imposes no restrict on
the Right-Hand-Side (RHS). In this case, constraint (4) ensures that $%
\varepsilon _{i,v,u}=1$ for every incident edge $e_{v,u}$ with $\pi
_{i,d,v,u}$ as 1. Therefore, $u$ is multicast capable for $T_{i}$, and each
packet is delivered once in every incident link.

The last two constraints are capacity constraints. Constraint (6) states
that each node $u$ can act as a branch state node of at most $b_{u}$ trees
in $\mathcal{T}$, while constraint (7) describes that the total multicast
bandwidth consumption of in each directed edge $e_{v,u}$ cannot exceed $%
c_{u,v}$.

%%%%%%%%%%%%%%%%%%%%%%%%%%%%%%%%%%%%%%%%%%%%%%%%%%%%%%%%%%%

\section{Hardness Results\label{sec: HardnessResults}}

In the following, we first show that SMTE-N is very challenging in
Complexity Theory by proving that it is NP-Hard and not able to be
approximated within $\delta ^{c}$ for every $c<1$, where $\delta
=\max_{1\leq i\leq t}|D_{i}|$. Afterward, we prove that SMTE cannot be
approximated within any ratio.

The Steiner tree problem is a special case of SMTE-N. However, SMTE-N is
much more challenging than the Steiner tree problem because the Steiner tree
problem can be approximated within ratio $1.55$ and is thus in APX in
Complexity Theory. In contrast, we find out that SMTE is much more difficult
to be approximated. The following theorem first proves that SMTE-N cannot be
approximated within $\delta ^{c}$ for every $c<1$, where $\delta
=\max_{1\leq i\leq t}|D_{i}|$, by a gap-introducing reduction from the 3SAT
problem.

\begin{theorem}
\label{hardness} For any $\epsilon >0$, there exists no $(\delta
^{1-\epsilon })$-approximation algorithm for SMTE-N, where $\delta
=\max_{1\leq i\leq t}|D_{i}|$, assuming P $\neq $ NP.
\end{theorem}

\begin{proof}
We prove the theorem with the gap-introducing reduction from the 3SAT
problem.

The 3SAT problem is a simplification of the regular SAT problem. An instance
of 3SAT is a conjunctive normal form (CNF) in which each clause contains
exactly three variables. The 3SAT problem is to decide, given a Boolean
expression $\phi $ in CNF such that each clause contains exactly three
variables, whether $\phi $ is satisfiable.

For any instance $\phi $ of the 3SAT problem, we build an instance $G(V,E)$
of SMTE-N with two multicast trees, where the destination sets are $D_{1}$
and $D_{2}$. Let $\mathrm{OPT}(G)$ denote the optimal solution of $G$ for
SMTE-N. The goals of the reduction are two-fold. 1) If $\phi $ is satisfiable then $\mathrm{OPT}(G)\leq 4p^{q+1}$. 2) If $\phi$ is not satisfiable then $\mathrm{OPT}(G)>(4p^{q+1})\times(\max \{|D_1|,|D_2|\})^{1-\epsilon}$. In the above two goals, $n$ is the number of Boolean variables in $\phi $, $m$ is the number of clauses in $\phi $, $p=\max \{m,n\}$ and $q$ is a large number (derived later).

To achieve the above goals, we build the instance of SMTE-N from each
instance of the 3SAT problem as follows. Given an instance $\phi $ of 3SAT
with $n$ Boolean variables $x_{1},\ldots ,x_{n}$ and $m$ clauses $%
C_{1},\ldots ,C_{m}$, we construct a directed graph $G(V,E)$ in the following way. 1) The node set $V$ is partitioned into four node sets $\{s\}$, $U$, $D_1$, and $D_2$. 2) $U$ includes $2n$ nodes $u_{1},\overline{u_{1}},u_{2},\overline{u_{2}},\ldots ,u_{n},\overline{u_{n}}$ (nodes $u_{i}$ and $\overline{u_{i}}$ correspond to the Boolean variable $x_{i}$), and for each $i$ with $1\leq i\leq n$, there are directed edges $(s,u_{i})$ and $(s,\overline{u_{i}})$. 3) $D_{1}$ has $mp^{q}$ nodes $d_{j}^{(k)}$, where $1\leq j\leq m$ and $1\leq k\leq p^{q}$ (nodes $d_{j}^{(k)}$, $1\leq k\leq p^{q}$, corresponding to $p^{q}$ copies of the clause $C_{j}$), and there exists a directed edge $(u_{i},d_{j}^{(k)})$ ($(\overline{u_{i}},d_{j}^{(k)})$) if and only if the variable $x_{i}$ ($\overline{x_{i}}$, resp.) appears in $C_{j}$. 4) $D_2$ contains $np^q$ nodes $w_i^{(k)}$, where $1\le i\le n$ and $1\le k\le p^q$, and there are directed edges $(u_{i},w_i^{(k)})$ and $(\overline{u_{i}},w_i^{(k)})$ for each $i,k$ with $1\le i\le n$ and $1\le k\le p^q$. Note that $G$ only has the directed edges described above, $p=\max \{m,n\}$, and $q$ is the smallest integer such that $q\geq(3+\log _{p}4)/\epsilon $. Fig.~\ref{fig:reduction_3SAT} presents an illustrative example of an SMTE-N instance.

The cost of each edge from $s$ to $U$ is set as $p^{q}$, and the cost of the
other edges are set to be $1$. The capacity of each node is set as $1$. Let $%
s$ and $D_{1}$ be the source node and destination set of $T_{1}$
respectively, and let $s$ and $D_{2}$ be the source node and destination set
of $T_{2}$ respectively.

If $\phi $ is satisfiable, there is a truth assignment to $x_{i}$ such that $%
\phi $ is true. Let $A=\{u_{i}:\text{$x_{i}$ is assigned to be true}\} \cup \{%
\overline{u_{i}}:\text{$x_{i}$ is assigned to be false}\}$. Consider the
tree $T_{1}$ rooted at $s$ that includes 1) the edges between $s$ and $A$, and 2) the edges between $d_{j}^{(k)}$ and
one of its neighbor in $A$ (the existence of its neighbor in $A$ comes from
that $\phi $ is satisfiable). Consider the tree $T_{2}$ that includes 1) the
edges between $s$ and $U\setminus A$, 2) the edges between $U\setminus A$ and $D_{2}$. Then
$(T_{1},T_{2})$ is a feasible solution of SMTE-N, and it can act as an upper
bound of SMTE-N in $G$. The total edge cost of $T_{1}$ is $np^{q}+mp^{q}$,
and the total edge cost of $T_{2}$ is $np^{q}+np^{q}$. Since the node
capacity is sufficient, the total bandwidth cost of this feasible solution
is $3np^{q}+mp^{q}\leq 4p^{q+1}$. Hence, we have $\mathrm{OPT}(G)\leq
4p^{q+1}$.

On the other hand, if $\phi $ is not satisfiable, let $(T_{1},T_{2})$ be any
feasible solution. For $1\leq k\leq p^{q}$, let $I_{k}$ be the set
consisting of every $i$ with $1\leq i\leq n$, such that $u_{i}$ and $%
\overline{u_{i}}$ are adjacent to some nodes in $\{d_{j}^{(k)}:1\leq j\leq
m\}$ along the edges in $T_{1}$. Since $\phi $ is not satisfiable, $I_{k}$
is not empty for each $k$ with $1\leq k\leq p^{q}$. By pigeonhole principle~\cite{cameron1994combinatorics}, 
there exists at least one $i^{\ast }$ with $1\leq
i^{\ast }\leq n$ such that $i^{\ast }$ is in at least $\frac{p^{q}}{n}$ sets
of $\left \{ I_{1},I_{2},...,I_{p^{q}}\right \} $. In $T_{1}$, therefore, $u_{i^{\ast
}}$ has at least $\frac{p^{q}}{n}\geq p^{q-1}$ downstream destination nodes, and $\overline{u_{i^{\ast }}}$ has at least $\frac{p^{q}}{n}\geq p^{q-1}$ downstream  destination nodes. On the other hand, in $T_{2}$, $u_{i^{\ast }}$ and $\overline{u_{i^{\ast }}}$ need to dominate $p^{q}$ downstream destination nodes of $T_{2}$. If $u_{i^{\ast }}$ or $\overline{u_{i^{\ast }}}$ is not a branch state node in $T_{1}$, then  the total cost is at least $p^{2q-1}$. On the other hand, if $u_{i^{\ast }}$ and $\overline{u_{i^{\ast }}}$ both are branch state nodes in $T_{1}$, since the capacity of node $u_{i^{\ast }}$ and $\overline{u_{i^{\ast }}}$ are 1, neither $u_{i^{\ast }}$ nor $\overline{u_{i^{\ast }}}$ are branch state nodes in $T_{2}$. The total cost is at least $p^{2q}$.

Therefore, the total cost of the optimal solution is larger than $p^{2q-1}$,
and we have $\mathrm{OPT}(G)>p^{2q-1}=(4p^{q+1})(p^{q-2-\log
_{p}4})=(4p^{q+1})(p^{q+1})^{\frac{q-2-\log _{p}4}{q+1}%
}=(4p^{q+1})(p^{q+1})^{1-\frac{3+\log _{p}4}{q+1}}\geq
(4p^{q+1})(p^{q+1})^{1-\epsilon }\geq (4p^{q+1})(\max
\{|D_{1}|,|D_{2}|\})^{1-\epsilon }$. Since $\epsilon $ can be arbitrarily
small, for any $\epsilon >0$, there is no $(\max
\{|D_{1}|,|D_{2}|\})^{1-\epsilon }$ approximation algorithm for SMTE-N,
assuming P $\neq $ NP. The theorem follows.
\end{proof}

\begin{figure}[t]
\centering
\includegraphics[width=2.3 in]{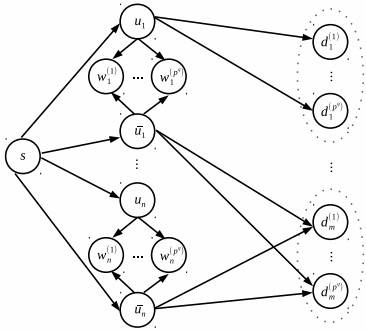} 
\vspace{-5pt}% \vspace{-15pt}
\caption{An illustration of instance building from 3SAT to SMTE-N}
\label{fig:reduction_3SAT}
\end{figure}

%%%%%%%%%%%%%%%%%%%%%%%%%%%%%%%%%%%%%%%%%%%%%%%%%%%%%%%%%%%%%%%%%%%%%%%%%%%%%%%
In the following, we prove that SMTE cannot be approximated within any ratio.

\begin{theorem}
For any polynomial time computable function $f$, SMTE cannot be approximated
within a factor of $f(|V|)$, unless P = NP. In other words, for arbitrary
positive integer $k$, SMTE cannot be approximated within $|V|^{k}$.
\end{theorem}

%\textbf{****** Proof of Theorem 2 ******}

\begin{proof}
Assume, for a contradiction, that there is a polynomial time approximation
algorithm $\mathcal{A}$ with the approximation raio  $f(|V|)$ for SMTE. This
proof will show that $\mathcal{A}$ can be used for deciding the 3SAT problem
in polynomial time, thus implying P = NP. 

Specifically, given a graph $G$, let $\mathrm{OPT}(G)$ denote the optimal
solution of $G$ for SMTE. For any instance of the 3SAT problem, we build an
instance $G(V,E)$ for SMTE with two multicast trees with the destination
sets $D_{1}$ and $D_{2}$. The goals of the reduction are two-fold:

\begin{enumerate}
\item if $\phi $ is satisfiable then $\mathrm{OPT}(G)\leq m + 3n$, and

\item if $\phi $ is not satisfiable then $\mathrm{OPT}(G) > (m + 3n) \times
f(|V|)$,
\end{enumerate}

where $n$ is the number of Boolean variables, $m$ is the number of clauses,
and $f$ is a polynomial-time computable function. 
%Hence there is no $f(|V|)$ approximation algorithm for SMTE unless P $= $ NP.

To achieve the above goals, we build the instance of SMTE from each instance
of the 3SAT problem. Given an instance $\phi $ of 3SAT with $n$ Boolean
variables $x_{1},\dots ,x_{n}$ and $m$ clauses $C_{1},\dots ,C_{m}$, we
construct a directed graph $G(V,E)$ such that

\begin{enumerate}
\item the node set $V$ is partitioned into four node sets $s$, $U$, $D_{1}$,
and $D_{2}$;

\item $U$ contains $2n$ nodes $u_{1}, \overline{u_{1}}, u_{2}, \overline{%
u_{2}}, \dots, u_{n}, \overline{u_{n}}$ (nodes $u_{i}$ and $\overline{u_{i}}$
are corresponding to the Boolean variable $x_{i}$), and for each $i$ with $1
\leq i \leq n$, there are directed edges $(s, u_{i})$ and $(s, \overline{%
u_{i}})$;

\item $D_{1}$ contains $m$ nodes $d_{1},\dots ,d_{m}$ (node $d_{j}$
corresponds to the clause $C_{j}$), and there exists a directed edge $%
(u_{i},d_{j})$ ($(\overline{u_{i}},d_{j})$) if and only if the variable $%
x_{i}$ ($\overline{x_{i}}$, resp.) appears in $C_{j}$;

\item $D_{2}$ contains $n$ nodes $d^{\prime}_{1}, \dots, d^{\prime}_{n}$ and
for each $i$ with $1 \leq i \leq n$, there are directed edges $(s,
d^{\prime}_{i})$, $(u_{i}, d^{\prime}_{i})$, and $(\overline{u_{i}},
d^{\prime}_{i})$;

\item $G$ only has the directed edges described above.
\end{enumerate}

The cost of each edge from $s$ to $D_{2}$ is set as $(m+3n)\times f(|V|)$ ,
and the cost of every other edge is set to be 1. The capacity of each
directed edge is set to be 1, and the data rate of each tree is also 1. The
node capacity is set as 2.

Let $s$ and $D_{1}$ be the source node and destination set of $T_{1}$,
respectively Let $s$ and $D_{2}$ be the source node and destination set of $%
T_{2}$ respectively.

If $\phi $ is satisfiable, there is a truth assignment to $x_{i}$ such that $%
\phi $ is true, let $W=\{u_{i}:x_{i}$ is assigned to be true$\} \cup \{%
\overline{u_{i}}:x_{i}$ is assigned to be false$\}$. Consider the tree $T_{1}
$ rooted at $s$ including 1) the edges between $s$ and $W$ and 2) the edge
between each $d_{j}$ and one of its neighbor in $W$ (the existence of its
neighbor in $W$ comes from that $\phi $ is satisfiable). Consider the tree $%
T_{2}$ which includes 1) the edges between $s$ and $U\setminus W$ and 2) the
edges between $U\setminus W$ and $D_{2}$. Then $(T_{1},T_{2})$ is a feasible
solution of SMTE and it can act as an upper bound of SMTE in $G$. The total
edge cost of $T_{1}$ is $m+n$ and the total edge cost of $T_{2}$ is $2n$.
Since the node capacity is sufficient, the total bandwidth cost of this
feasible solution is $m+3n$. Hence $\mathrm{OPT}(G)\leq m+3n$.

On the other hand, if $\phi $ is not satisfiable, let $(T_{1},T_{2})$ be any
feasible solution. Since $\phi $ is not satisfiable, there is an $i$ such
that both edges $(s,u_{i})$ and $(s,\overline{u_{i}})$ appear in $T_{1}$ for
any feasible solution of SMTE, in order to span all destinations in $D_{1}$.
Therefore, the edges $(s,u_{i})$ and $(s,\overline{u_{i}})$ cannot be
included in $T_{2}$ due to the link capacity constraint, and $T_{2}$ thereby
needs to choose the directed edge $(s,d_{i}^{\prime })$. The total edge cost
of $T_{2}$ is at least $(m+3n)\times f(|V|)$ in this case, and the total
bandwidth cost of the optimal solution in SMTE is greater than $(m+3n)\times
f(|V|)$. Therefore SMTE cannot be approximated within a factor of $f(|V|)$,
unless P = NP.
\end{proof}

\section{Algorithm Design}

In the following, we first propose a $\delta $-approximation algorithm, named 
\textit{Multi-Tree Routing and State Assignment Algorithm (MTRSA)}, for
SMTE-N, where $\delta =\max_{1\leq i\leq t}|D_{i}|$. Note that we first
focus on the node capacity, instead of the link capacity, because the
scalability in Group Table is crucial in SDN and has not been explored in
previous studies of multicast tree routing for other networks. Since Theorem %
\ref{hardness} proves that there is no $(\delta^{1-\epsilon })$-approximation
algorithm of SMTE-N for any $\epsilon >0$, MTRSA achieves the best
approximation ratio. Afterward, we extended it to support SMTE. 

%Due to the
%space constraint, the pseudo-code is presented in \cite{extension}.

\subsection{Algorithm Description}
\label{sec:aglo_des}

MTRSA includes two phases: 1) Multi-Tree Routing Phase and 2) State-Node
Assignment Phase. Multi-Tree Routing Phase first constructs an initial
multicast tree for each multicast group to minimize the total bandwidth
consumption and balance the distribution of branch nodes in different trees.
State-Node Assignment Phase then finds the branch state nodes for each
multicast tree to follow the node constraint.

\subsubsection{Multi-Tree Routing Phase}

Initially, Multi-Tree Routing Phase constructs a shortest-path tree with
source $s_{i}$ and destination set $D_{i}$ for each tree $T_{i}\in \mathcal{T%
}$. A node $u$ is \emph{full} if it acts as a branch node for $b_{u}$
multicast trees. By contrast, $u$ is \emph{overloaded} if it acts as a
branch node for more than $b_{u}$ trees. In this case, $u$ needs to act as a
branch stateless node for some of those trees and thereby will incur more
bandwidth consumption. To address this issue, after finding the
shortest-path trees, if there is an overloaded node, we adjust the local
tree routing nearby the overloaded node to move the branch node to another
node that has not been full, in order to balance the distribution of branch
nodes among different multicast trees.

More specifically, if any node $u$ is an overloaded node and a branch node
in any tree $T_{i}$, MTRSA chooses a node $v$ of $T_{i}$ such that: 1) $v$
is a downstream to $u$ in $T_{i}$, 2) $v$ is a branch node or a destination
node of $T_{i}$, and 3) there is no other branch node or destination node in
the path from $u$ to $v$ in $T_{i}$. In other words, $v$ is a nearby
downstream branch node of $u$ and a destination node. MTRSA reroutes the path (from $u$ to $v$) to
another path (from $w$ to $v$) as follows, in order to alleviate the storage
load in $u$. Let $\ell $ denote the total bandwidth cost of the path from $u$
to $v$ in $T_{i}$. We find a new path from $w$ to $v$ such that: 1) the
total cost of the new path is at most $\ell $, 2) the new path does not pass
through any exiting node in $T_{i}$, and 3) this new path starts from an
on-tree node $w$ such that i) it is not a leaf node of $T_i$, and ii) it is not full or overloaded. We update
tree $T_{i}$ by substituting the old path from $u$ to $v$ in $T_{i}$ with
the new path from $w$ to $v$, and the overload situation in $u$ can be
alleviated accordingly. Afterward, we process every other downstream branch
node $v$ of $u$ until $u$ is no longer a branch node for $T_{i}$. The above
process is repeated for every tree $T_{i^{\prime }}$ iteratively until $u$
is no longer overloaded.

\textbf{Example. } Consider the following example in Fig. \ref{fig2:subfig:a}%
. Let $G(V,E)$ be the network with two multicast trees $T_{1}$ and $T_{2}$ with the data rate as 1.
The number on each edge is the unit bandwidth cost of this edge, and the node
capacity of each node is $1$. The source $s_{1}$ of the first tree $T_{1}$
is $s$ with the corresponding destination set $D_{1}=\{d_{1},d_{2},\dots
,d_{8}\}$, while the source $s_{2}$ of $T_{2}$ is also $s$, but the
destination set is $D_{2}=\{d_{1}^{\prime },d_{2}^{\prime },\dots
,d_{6}^{\prime }\}$. In Multi-Tree Routing Phase, we first find the blue and
red shortest-path trees $T_{1}$ and $T_{2}$ in Fig. \ref{fig2:subfig:b}.
Afterward, we adjust the multicast trees for overloaded nodes. Specifically,
the node capacity of $a$ is $1$, but $a$ is a branch node of both $T_{1}$
and $T_{2}$. Therefore, $a$ is an overloaded node, and MTRSA examines nodes $%
d_{1},d_{2},v,c$, which are downstream nodes of $a$ in $T_{1}$. MTRSA first
reroutes the path $\left \{ a,b,c\right \} $ in tree $T_{1}$. Since node $y$
is overloaded, node $c$ cannot be rerouted from node $y$. In contrast, node $%
v$ is a full branch node of $T_{1}$, and MTRSA reroutes node $c$ from node $v$
for $T_{1}$ as shown in Fig. \ref{fig2:subfig:c}. Note that the bandwidth
cost is efficiently reduced since the new path from $v$ to $c$ is much smaller than
the one from $a$ to $c$. Therefore, Multi-Tree Routing Phase addresses both
the node capacity and the bandwidth consumption for scalable multicast
traffic engineering.

\subsubsection{State-Node Assignment Phase}

It is worth noting when the network is heavily loaded, the first phase may
not be able to ensure that every overloaded node can be successfully
adjusted to balance the distribution of branch nodes in different trees, and
State-Node Assignment Phase is crucial in this case to minimize the
increment of bandwidth consumption due to unicast tunneling through branch
stateless node. More specifically, State-Node Assignment Phase includes two
stages: 1) Greedy Assigning Stage, and 2) Local Search Stage. Greedy
Assigning Stage assigns the branch state nodes by iteratively maximizing the
reduction of the number of branch state nodes, and later in Section \ref%
{sec: AppRatio} we prove that the number of branch state nodes reduced
by the Greedy Assigning Stage is at least half of the number of branch state
nodes reduced by an optimal strategy. Local Search Stage then improves the
solution by further alleviating the assignment on overloaded nodes and
rerouting the trees to further reduce the total bandwidth cost. We detail
the two stages as follows.

For each multicast tree $T_{i}$ obtained in Multi-Tree Routing Phase, let $%
W_{i}$ denote the set of branch nodes in $T_{i}$, and $W=\cup _{1\leq i\leq t}W_{i}$. On the other hand,
let $A_{i}$ be the set of branch state nodes in $T_{i}$ to be decided in
this phase, and $A_{i}$ thereby is a subset of $W_{i}$. Let $c(T_{i},A_{i})$
denote the total bandwidth cost of $T_{i}$ with the set of branch state
nodes as $A_{i}$. More precisely, $c(T_{i},A_{i})=\sum_{v\in A_{i}\cup
D_{i}}c(P_{v})$, where $P_{v}$ is the path from the closest upstream branch
state node in $A_{i}$ or the source to $v$, such that all internal nodes of $%
P_{v}$ are not in $A_{i}$, and $c(P_{v})$ is the cost of all edges in $P_{v}$%
. In other words, if there is no branch stateless node in $P_{v}$, every
packet is delivered only once on every link of $P_{v}$. By contrast, if $%
P_{v}$ includes a branch stateless node $u$, each packet is sent multiple
times on the links from the closest upstream branch state node to $u$,
corresponding to the unicast tunneling case.

An assignment $A$ of branch state nodes can be defined as follows: $A$ is a $%
0,1$-matrix with the rows indexed by $\{1,\ldots ,t\}$ and columns indexed
by $W$, such that 1) the $1$'s in row $i$ can only be the columns indexed in 
$W_{i}$, and 2) the number of $1$'s in column $w\in W$ is no more than the
node capacity $b_{w}$. We assign a branch state node $w\in W$ to tree $T_{i}$
if and only if the $(i,w)$ entry of $A$ is $1$. In other words, the first
condition ensures that a branch state node can only be assigned to a branch
node of $T_{i}$, while the second condition is the node capacity constraint.
Given an assignment $A$ of branch state nodes, let $A_{i}=\{w\in W:\text{the 
$(i,w)$ entry of $A$ is $1$}\}$ denote the set of branch state nodes for $%
T_{i}$, and the total bandwidth cost for the set $\mathcal{T}$ of all
multicast trees with the state-node assignment $A$ is $c(\mathcal{T}%
,A)=\sum_{1\leq i\leq t}c(T_{i},A_{i})$. Since an assignment $A$ of branch
state nodes can also be regarded as a subset of $N$, where $N=\{1,\ldots
,t\} \times W$, let $\mathcal{M}$ be the family of subsets of $N$ satisfying
the above two conditions (hence $\mathcal{M}$ is the family of all feasible
assignments of branch state nodes to $\mathcal{T}$), and we use $c(\mathcal{T%
},\varnothing )$ to denote the total bandwidth cost of $\mathcal{T}$ without
assigning any branch state node. Now let the set function $z:\mathcal{M}%
\rightarrow \mathbb{R}$ such that $z(A)$ represents the cost reduced by
assignment $A$. More formally, $z(A)=c(\mathcal{T},\varnothing )-c(\mathcal{T%
},A)$ for each $A\in \mathcal{M}$.

The above matrix representation plays a crucial role in Greedy Assigning
Stage when we prove the quality of the state-node assignment based on Matroid
Theory later in Section \ref{sec: AppRatio}. This stage starts from a branch
state node assignment $\varnothing $ and cost $c(\mathcal{T},\varnothing )$,
and iteratively assigns one branch state node for a tree in $\mathcal{T}$
until no more assignment can reduce $c(\mathcal{T},A)$. More precisely, in
each iteration, if the present branch state node assignment is $A\in 
\mathcal{M}$, we choose an element $x$ in $N-A$ such that: 1) $A\cup
\{x\}$ is in $\mathcal{M}$ and follows the node capacity constraint, and 2) $%
z(A\cup \{x\})=\max_{y\in (N-A)}z(A\cup \{y\})$. In other words, the first
condition guarantees that the new assignment is feasible, whereas the second
condition chooses the node leading to the maximal reduction on $c(\mathcal{T}%
,A)$. Afterward, Local Search Stage first adjusts the assignment of branch
state nodes for overloaded nodes iteratively. In each iteration, we first
extract an overloaded node $u$ and then compute the reduction of the
bandwidth cost with a branch state node assigned to $u$ for each tree $T$
spanning $u$, assuming that the state-node assignment of other nodes are not
changed. Afterward, this phase sorts the trees according to the bandwidth
reduction and chooses the $b_{u}$ trees with the largest reduction, whereas
the branch state nodes are assigned to them accordingly. This stage is
repeated until all overloaded nodes are carefully examined. Afterward, this
stage reroutes the paths from other branch nodes of a tree in order to find
a smaller tree with the same assignment of branch state nodes. More
specifically, for any branch node $u$ in tree $T_{i}$, we choose nodes $v$
and $w$ of $T_{i}$ in the same way as the Multi-Tree Routing Phase in order
to find a new path from $w$ to $v$, and $w$ is not full.

\textbf{Example.} In Greedy Assignment Stage of the State-Node Assignment
Phase, when there is no branch state node, the total bandwidth cost in Fig. %
\ref{fig2:subfig:c} is $c(\mathcal{T},\varnothing )=c(T_{1},\varnothing
)+c(T_{2},\varnothing )=142+92=234$. If we assign a branch state node on $v$
for tree $T_{1}$, the bandwidth cost of the path $s,a,u,v$ can be reduced by 
$(4-1)$ times since there are $4$ downstream destination nodes $%
d_{3},d_{4},d_{5},d_{6}$ of $v$ in $T_{1}$. The reduced cost is the largest
among all possible branch state node assignments. Therefore, MTRSA first
assigns a branch state node on $v$ for tree $T_{1}$ with the cost reduced by 
$(4-1)\times (k_{s,a}+k_{a,u}+k_{u,v})=63$. It then assigns a branch state
node on $y$ to $T_{2}$ with the cost reduced by $(3-1)\times
(k_{s,w}+k_{w,y})=30 $. Afterward, node $a$ is assigned as a branch state
node for $T_{2}$ with the cost reduced by $(3-1)\times (k_{s,a})=18$, and
node $c$ is assigned as a branch state node for $T_{2}$ with the cost reduced by 
$(2-1)\times (k_{a,b}+k_{b,c})=7$. In Local Search Stage, there are three
overloaded nodes $a$, $c$, and $y$. For overloaded node $a$, this phase
moves the branch state node on $a$ from $T_{2}$ to $T_{1}$ without changing
the branch state nodes of the other nodes. If we assign a branch state node
on $a$ to $T_{1}$, it becomes possible to reduce the cost of $T_{1}$ by $%
(3-1)\times (k_{s,a})=18$. In contrast, if we assign a branch state node on $%
a$ to $T_{2} $, we are able to reduce the cost of $T_{2}$ by only $(2-1)\times
(k_{s,a})=9$. Nodes $c$ and $y$ are then processed similarly.\textbf{\ }%
Finally, in Fig. \ref{fig2:subfig:d}, since node $y$ has been a branch state
node, node $c$ can be re-routed to node $y$ in $T_{2}$, and the total
bandwidth consumption is reduced from $107$ in Fig. \ref{fig2:subfig:c} to $%
93$ in Fig. \ref{fig2:subfig:d} accordingly.

\begin{figure}[t]
\centering
\subfigure[Original network]{
    \label{fig2:subfig:a}
    \includegraphics[height = 3cm, width = 4cm]{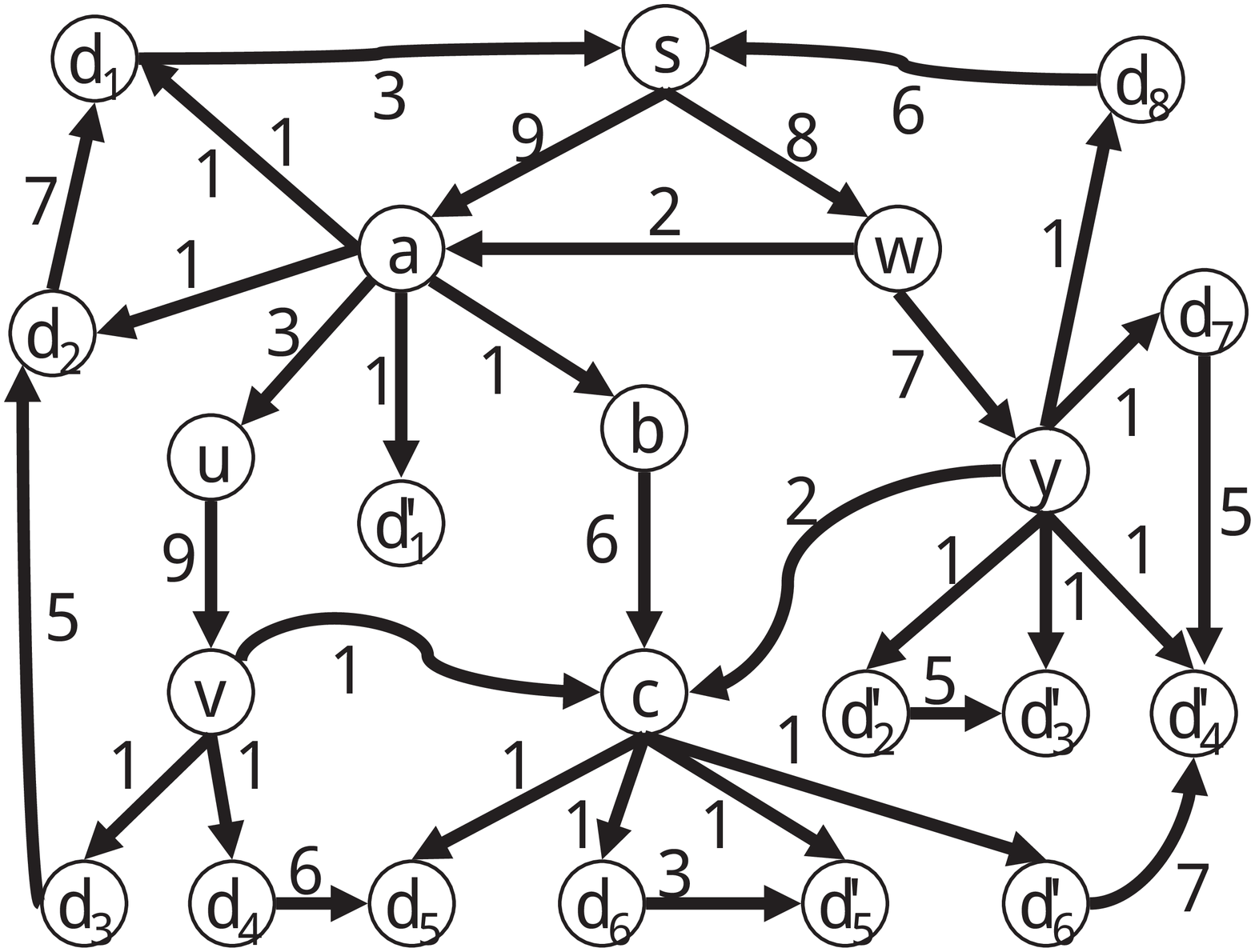}} \smallskip 
\subfigure[Shortest-path trees]{
    \label{fig2:subfig:b}
    \includegraphics[height = 3cm, width = 4cm]{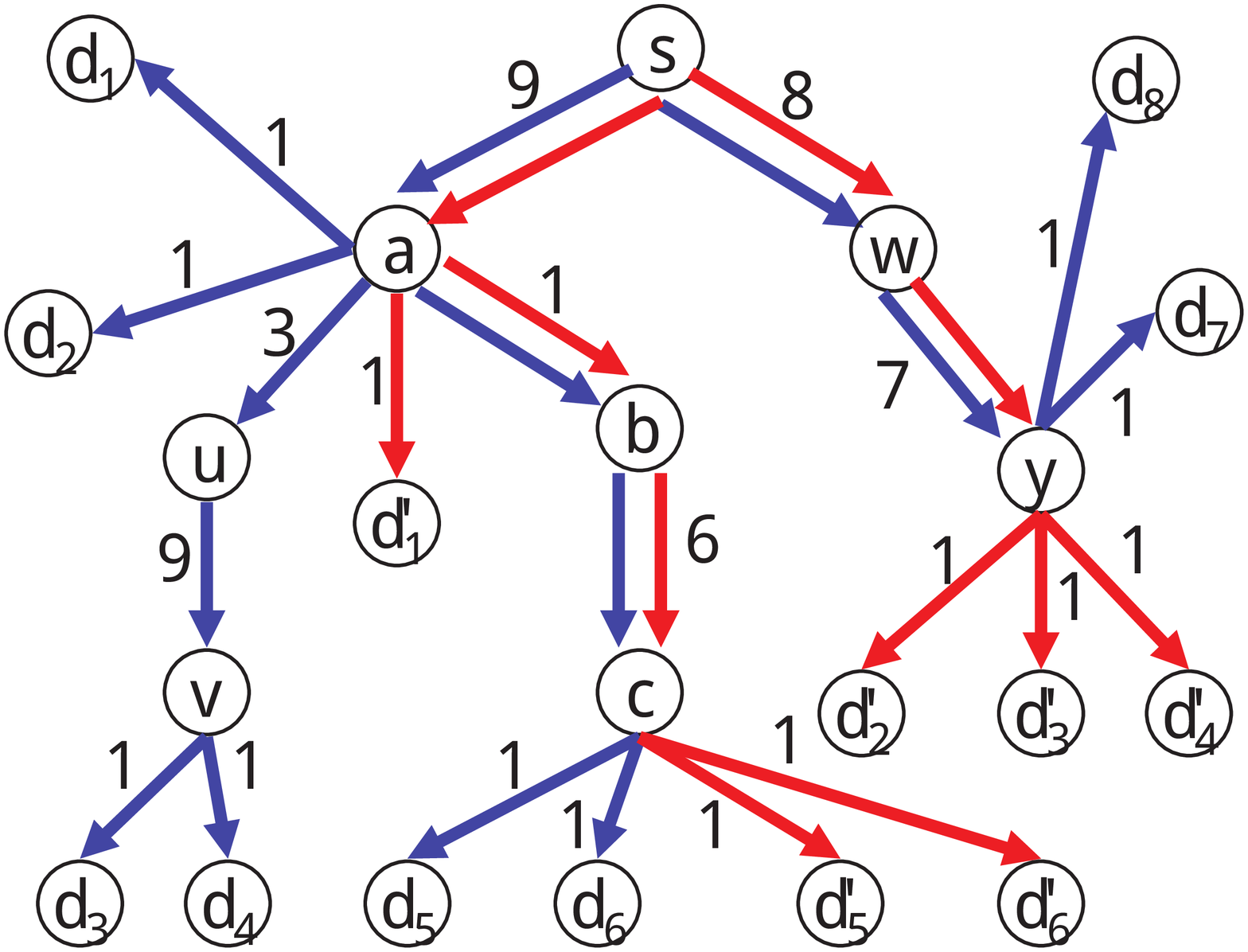}} \smallskip 
\subfigure[Multi-Tree Routing Phase]{
    \label{fig2:subfig:c}
    \includegraphics[height = 3cm, width = 4cm]{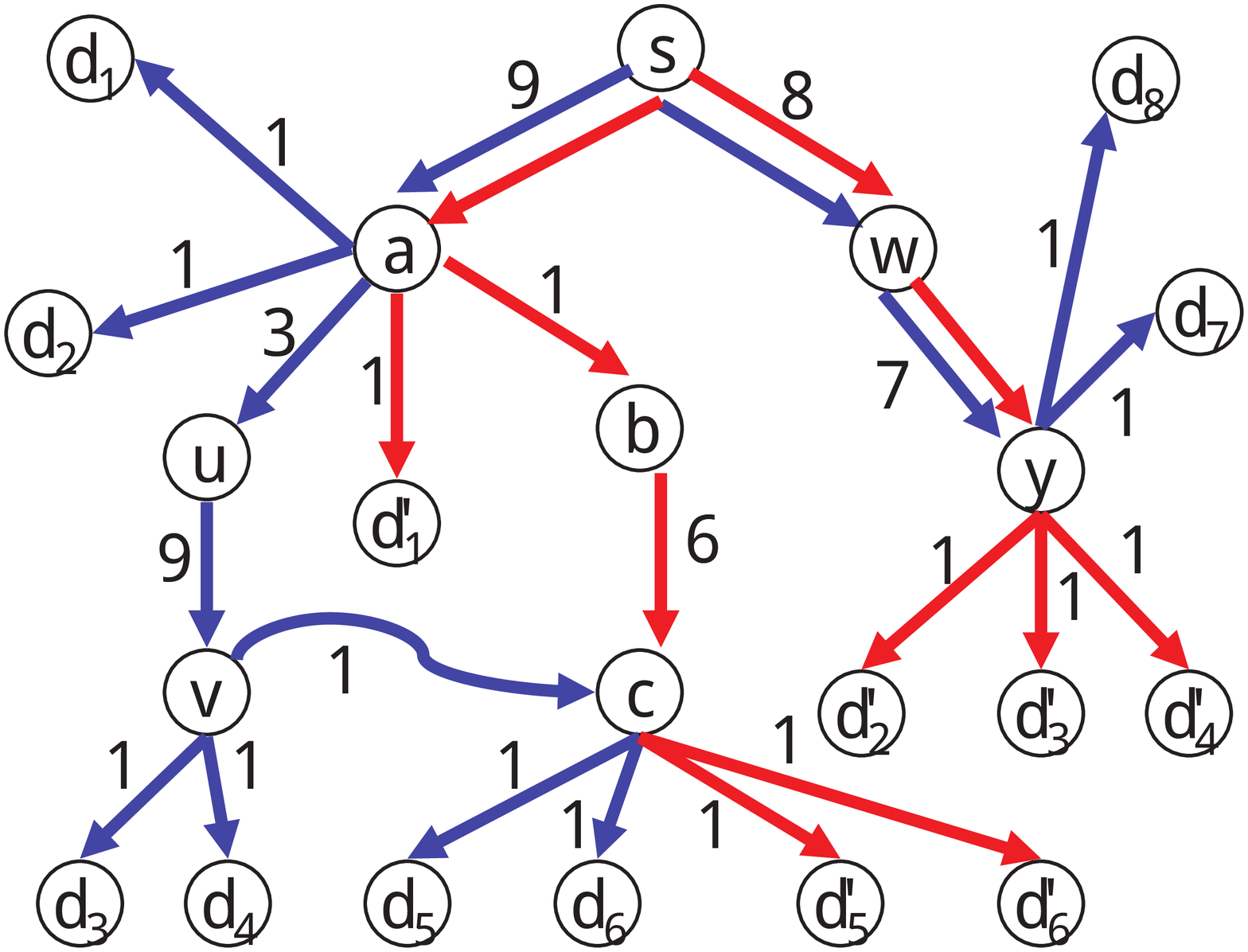}} \smallskip 
\subfigure[State-Node Assignment Phase]{
    \label{fig2:subfig:d}
    \includegraphics[height = 3cm, width = 4cm]{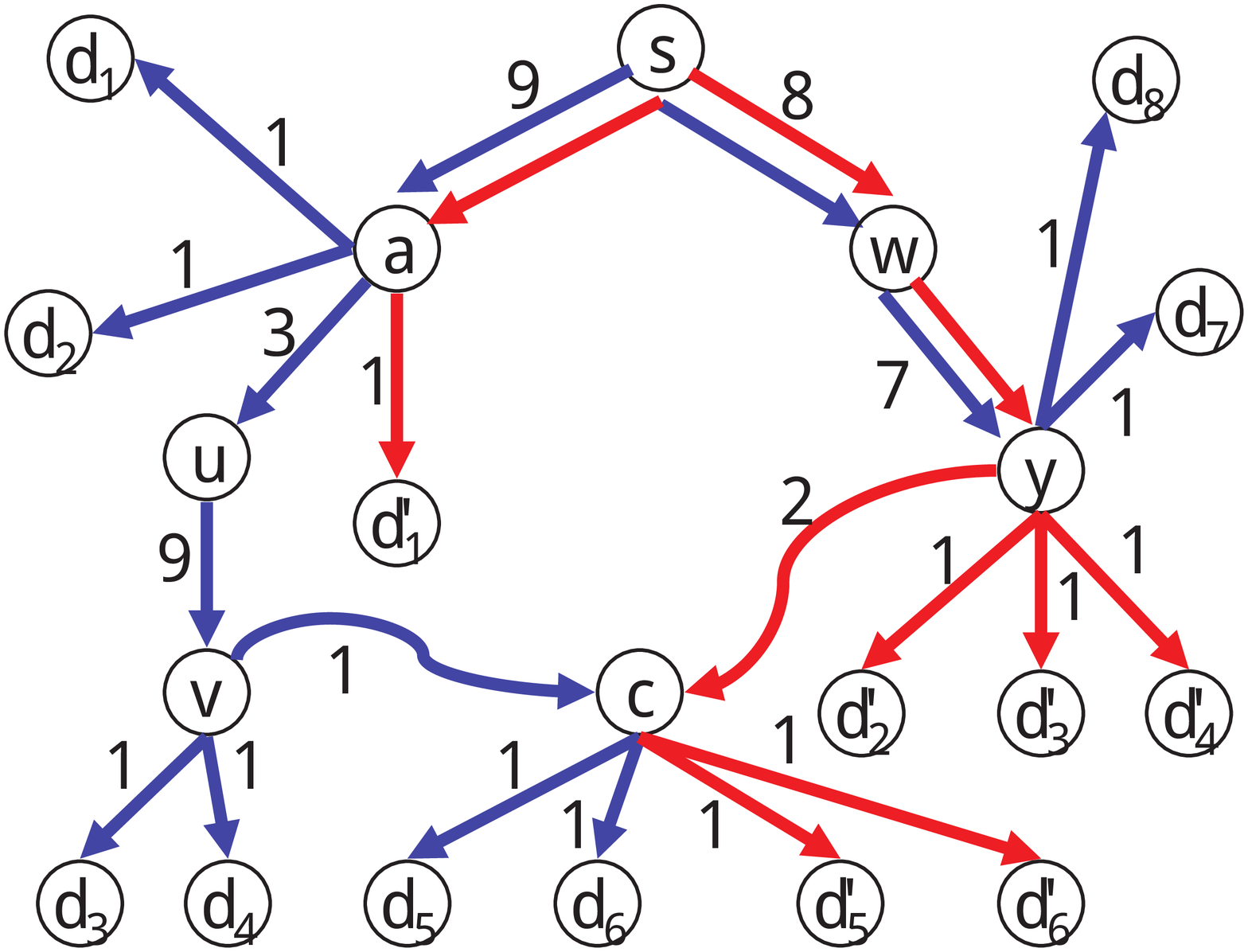}} \smallskip \vspace{%
-0.1in} %\vspace{-0.25in}
\caption{An example of MTRSA}
\label{fig2:subfig}
\end{figure}

\subsection{Approximation Ratio and Time Complexity\label{sec: AppRatio}}

In the following, we first examine the quality of assignment for branch
state nodes in the second phase. We prove that $(N,\mathcal{M})$ is a
matroid and the set function $z:\mathcal{M}\rightarrow \mathbb{R}$ is a
nondecreasing submodular set function. Therefore, according to the Matroid
Theorem for maximizing submodular set function~\cite{greedy78}, we
have the following theorem.

\begin{theorem}
\label{greedy assign} The number of branch state nodes reduced by the Greedy
Assignment Stage is at least one half of the branch state nodes reduced by
the optimal assignment of branch state nodes.
\end{theorem}

%%%%%%%%%  Extension Begin %%%%%%%%%%%%%%%%%%%%
\begin{proof}
In the following, we first prove that $\mathcal{M}$ is a matroid. $\mathcal{M%
}$ is the family of subsets of $N=\{1,\dots ,t\} \times W$ (i.e., we put the
elements of $N$ in a $t\times |W|$ array) such that the elements in the $i$%
-th row are in the columns indexed by $W_{i}$, and the number of elements in
the column indexed by $w$ is at most the capacity of $w$. Hence, by
definition, we have: 1) $\varnothing \in \mathcal{M}$, 2) If $A\subseteq
B\in \mathcal{M}$, then $A\in \mathcal{M}$, and 3) If $A,B\in \mathcal{M}$
with $|A|<|B|$, then there is an element $b\in B$ such that $A\cup \{b\} \in 
\mathcal{M}$. Therefore $\mathcal{M}$ is a matroid.

Now we prove that the set function $z:\mathcal{M}\subset 2^{N}\rightarrow 
\mathbb{R}$ is submodular and nondecreasing. Let $A,B\in \mathcal{M}$ with $%
A\subseteq B$ and $c\in N$ be the element in row $i$ and column $w$ such
that $A\cup \{c\},B\cup \{c\} \in \mathcal{M}$, since $z(A\cup \{c\})-z(A)$
is the cost reduced by assigning a branch state node in node $w$ to tree $%
T_{i}$ with branch state node assignment $A$, and since $z(B\cup \{c\})-z(B)$
is the cost reduced by assigning a branch state node in node $w$ to tree $%
T_{i}$ with branch state node assignment $B$, we have $z(A\cup
\{c\})-z(A)\geq z(B\cup \{c\})-z(B)$. Hence, $z$ is submodular. Let $A,B\in 
\mathcal{M}$ with $A\subseteq B$, by definition of $z$, we have $z(A)\leq
z(B)$, and $z$ thereby is nondecreasing.

Let $Z_{OPT}$ be $\max \{z(A):A\in \mathcal{M}\}$ and $Z_{G}$ be the result
from our algorithm. By a result on maximizing submodular set function on
matroid~\cite{greedy78}, we have $\frac{Z_{OPT}-Z_{G}}{Z_{OPT}-z(\varnothing )}%
\leq \frac{1}{2}$. Hence, $Z_{G}\geq \frac{1}{2}Z_{OPT}$. The theorem
follows.
\end{proof}

%%%%%%%%%  Extension Begin %%%%%%%%%%%%%%%%%%%%

%\begin{proof}
%Due to the space constraint, the detailed proof is presented in \cite{extension}.
%\end{proof}

Then, we prove that MTRSA is a $\delta $-approximation algorithm
for SMTE-N, where $\delta $ is the maximum size of the destination sets.
Since Theorem \ref{hardness} proves that there is no approximation algorithm
with ratio $\delta ^{1-\epsilon }$ for any $\epsilon >0$, the following
theorem shows that MTRSA achieves the best approximation ratio. In contrast,
since SMTE cannot be approximate within any ratio unless $P=NP$, it is
impossible to derive an approximation ratio for any algorithm of SMTE, and
we thereby evaluate MTRSA for SMTE in Section~\ref{sec:evaluation}.

\begin{theorem}
\label{approx} MTRSA is a $\delta $-approximation algorithm for SMTE-N,
where $\delta =\max_{1\leq i\leq t}|D_{i}|$.
\end{theorem}

\begin{proof}
Let the set of multicast trees $\mathcal{T}^{\ast }=(T_{1}^{\ast },\dots
,T_{t}^{\ast })$ with the assignment $A^{\ast }$ of branch state nodes be
the optimal solution to SMTE-N, and $W^{\ast }=\cup _{i=1}^{t}W_{i}^{\ast }$
with $W_{i}^{\ast }$ as the set of branch nodes of $T_{i}^{\ast }$, whereas $%
A_{i}^{\ast }=\{w\in W^{\ast }:\text{the $(i,w)$ entry of $A^{\ast}$ is $1$}\}$ be
the set of branch state nodes in $T_{i}^{\ast }$. Therefore, the optimal
bandwidth cost is $c(\mathcal{T}^{\ast },A^{\ast
})=\sum_{i=1}^{t}c(T_{i}^{\ast },A_{i}^{\ast })$. For MTRSA, Multi-Tree
Routing Phase first constructs the shortest-path trees $\mathcal{T}%
^{(1)}=(T_{1}^{(1)},\dots ,T_{t}^{(1)})$, and the rerouting procedure of the
Multi-Tree Phase outputs the new trees  $\mathcal{T}^{(2)}$. MTRSA finally
generates the trees $\mathcal{T}^{(3)}$ with the assignment $A^{(3)}$ of 
branch state nodes. Let $A$ be an assignment such that for each
non-overloaded node $u$ within $\mathcal{T}^{(1)}$, MTRSA assigns a branch
state node on $u$ to each tree in $\mathcal{T}^{(1)}$ with $u$ as a branch
node. According to the State-Note Assignment Phase, we have $A\subseteq
A^{(3)}$. On the other hand, in Multi-Tree Routing Phase, MTRSA updates the
trees only when the bandwidth cost does not increase, and the branch state
nodes can only reduce the cost. Therefore, we have $c(\mathcal{T}%
^{(3)},A^{(3)})\leq c(\mathcal{T}^{(2)},A^{(3)})\leq c(\mathcal{T}^{(2)},A)$%
. In the rerouting procedure of Multi-Tree Routing Phase, suppose we reroute 
$\mathcal{T}$ to $\mathcal{T^{\prime }}$. Since each rerouting step does not
create any new overloaded node, and the new path in $\mathcal{T^{\prime }}$
ensures that $c(\mathcal{T^{\prime }},A)\leq c(\mathcal{T},A)$. Therefore, $%
c(\mathcal{T}^{(2)},A)\leq c(\mathcal{T}^{(1)},A)$ holds by induction, and
we have $c(\mathcal{T}^{(3)},A^{(3)})\leq c(\mathcal{T}^{(2)},A)\leq c(%
\mathcal{T}^{(1)},A)\leq c(\mathcal{T}^{(1)},\varnothing )$. Since the path $%
P_{s_{i},d}^{(1)}$ from the source $s_{i}$ to node $d$ in $D_{i}$ in tree $%
T_{i}^{(1)}$ is the shortest path from $s_{i}$ to $d$ in $G$, and $%
T_{i}^{\ast }$ has a path from $s_{i}$ to $d$, we have $c(P_{s_{i},d}^{(1)})%
\leq c(T_{i}^{\ast })$ for every $i$ and every $d\in D_{i}$. Therefore $c(%
\mathcal{T}^{(3)},A^{(3)})\leq c(\mathcal{T}^{(1)},\varnothing
)=\sum_{i=1}^{t}c(T_{i}^{(1)},\varnothing )\leq \sum_{i=1}^{t}\sum_{d\in
D_{i}}c(P_{s_{i},d}^{(1)})\leq \sum_{i=1}^{t}\sum_{d\in D_{i}}c({T}%
_{i}^{\ast },A_{i}^{\ast })=\sum_{i=1}^{t}|D_{i}|\times c({T}%
_{i}^{\ast },A_{i}^{\ast })\leq \delta(\sum_{i=1}^{t}c({T}_{i}^{\ast
},A_{i}^{\ast }))=\delta\times c(\mathcal{T}^{\ast },A^{\ast })$. The theorem
follows.
\end{proof}

\textbf{Time Complexity.} We first find the shortest path between any two
nodes in $G$ with Johnson's algorithm in $O(|V||E|+|V|^{2}\log |V|)$ time as the
pre-processing procedure. Multi-Tree Routing Phase constructs the
shortest-path tree for each source $s_{i}$ and its corresponding destination
set $D_{i}$, and MTRSA compares the distance from a destination node to all
other nodes in $O(|V|)$ time. Processing all $d\in D_{i}$ requires $%
O(|V||D_{i}|)=O(\delta |V|)$ time, and processing all shortest-path trees
requires $O(t\delta |V|)$ time. After constructing the shortest-path trees,
MTRSA reroutes the paths from each overloaded node to some of its downstream
nodes. Since there are at most $t\delta $ branches in the tree set, we
reroute at most $\delta $ paths for each branch node, and each rerouting
requires the comparison of at most $|V|$ distances of paths. Therefore, the
above procedure requires $O(t\delta ^{2}|V|\log |V|)$ time, and Routing
Phase requires $O(t\delta ^{2}|V|\log |V|)$ time.

Afterward, in Greedy Assigning Stage of State-Node Assignment Phase, there
are at most $t|V|$ branch state nodes required to be assigned, and this
stage has at most $t|V|$ iterations. In each iteration, we first derive
the minimum cost reduced by assigning a branch state node on each node $v$
of every $T_{i}$ in $O(t|V|)$ time, and we update the cost reduced by
assigning each new branch state node to the tree $T_{i}$ in $%
O(|T_{i}||E|)=O(|V||E|)$ time. Therefore, this stage requires $%
O(t|V|(t|V|+|V||E|))=O(t|V|^{2}(t+|E|))$ time. Then, Local Search Stage
carefully examines the overloaded nodes. In each iteration, we adjust the
branch state nodes on each node $u$ in different trees without changing
other branch state nodes to find the new bandwidth cost in $O(t|E|)$ time,
and this stage takes $O(t|V||E|)$ time because there are at most $|V|$
iterations. Therefore, State-Node Assignment Phase requires $%
O(t|V|^{2}(t+|E|))$ time to allocate the branch state nodes and $O(t\delta
^{2}|V|\log |V|)$ time for rerouting, and MTRSA requires $O(t\delta
|V|^{2}\log |V|(t+|E|))$ time.

\subsection{\protect \bigskip Extension to SMTE\label{sec: extension}}

For SMTE, since the number of times that each packet is delivered in a link
cannot be acquired before assigning the branch state nodes, we first present
the concept of \emph{weak edge capacity constraint}. Let $\varepsilon
_{i,u,v}^{\prime }=1$ if $\varepsilon _{i,u,v}$ is a positive integer in our
Integer Programming formulation, and $\varepsilon _{i,u,v}^{\prime }=0$
otherwise. MTRSA needs to ensure $\sum_{1\leq i\leq t}f_{i}\times
\varepsilon _{i,u,v}^{\prime }\leq c_{u,v}$ holds for $\forall e_{u,v}\in E$
before assigning the branch state nodes. In addition, in the general case of
SMTE, the storage size of each branch state node $u$ in Group Table is
proportional to the degree of $u$ in the corresponding multicast tree~\cite{mckeown_openflow:_2008}. Therefore, let $\beta _{i,u}$ denote the node
weight (i.e., storage size) of assigning a branch state node on $u$ to tree $%
T_{i}$, and $b_{u}$ here denotes the size of Group Table in $u$. For SMTE,
we extend MTRSA as follows.

%\begin{enumerate} 
Before Multi-Tree Routing Phase, we sort the multicast trees in $%
\mathcal{T}$ according to their data rates in the ascending order,  $%
f_{1}\leq f_{2}\leq \dots \leq f_{t}$. In Multi-Tree Routing Phase, we find the first shortest-path tree $%
T_{1}$ in $\mathcal{T}$ and decrease the link capacity $c_{u,v}$ for every
edge in $T_{1}$ by its flow rate $f_{1}$. Note that any edge $e_{u,v}$ with
insufficient residual capacity to support $f_{2}$ will be removed since it cannot
support the rest of the multicast flows. The above procedure is repeated for
other trees in $\mathcal{T}$. 

In the rerouting procedure of Multi-Tree Routing Phase, we reroute
each multicast tree $T_{i}$ according to the weak edge capacity constraint,
such that any new path from $w$ to $v$ in Section~\ref{sec:aglo_des} must have sufficient capacity to
support $f_{i}$. In Greedy Assignment Stage, we find an element $x=(i,u)$ in $N-A$
according to $\frac{z(A\cup \{(i,u)\})-z(A)}{\beta _{i,u}}=\max \{ \frac{%
z(A\cup \{(i^{\prime },u^{\prime })\})-z(A)}{\beta _{i^{\prime },u^{\prime }}%
}:A\cup \{(i^{\prime },u^{\prime })\} \in \mathcal{M}\}$, which represents
the normalized cost reduction. In other words, the node weight $\beta _{i,u}$
is considered during the assignment of branch state nodes.

In Local Search Stage, optimizing the state-node assignment of one
node becomes the same as the knapsack problem because each candidate branch
state node now has a profit (i.e., cost reduction) and a size (i.e., node
weight), and we exploit Polynomial-Time Approximation Scheme for
knapsack~\cite{Caprara2000333} to find the solution. In the rerouting procedure of Local Search Stage, since now the branch
state nodes have been specified, we reroute each multicast tree $T_{i}$
according to the original edge capacity constraint $(7)$, such that any new
path from $w$ to $v$ must have sufficient capacity to support $f_{i}$. If
the amount of multicast flows in any edge exceeds the capacity constraint,
we also reroute its closest upstream state node $u$ in a tree $T_{i}$ to $w$%
, such that the new path from $w$ to $v$ follows the link capacity
constraint.

MTRSA can support the dynamic multicast group membership as follows. When a user $v$ joins or leaves a multicast group, MTRSA adds or trims (if no other users are located downstream to the user) the corresponding branch from the upstream branch node $u$ in the same way as Multi-Tree Routing Phase. Afterward, State-Node Assignment Phase adjusts the new branch state node if necessary.  Therefore, it does not need to re-compute the whole tree.

\subsection{Pseudo Code}
The pseudo code of MTRSA is shown in Algorithm~\ref{alg}.

\begin{algorithm}[t]
%\color{red}
\caption{ Multi-Tree Routing and State Assignment Algorithm (MTRSA)}
\label{alg}
\begin{algorithmic}[1]
\REQUIRE{A network $G=(V,E)$, source nodes $s_{1}, s_{2}, \dots, s_{t}$, destination sets $D_{1}, D_{2}, \cdots, D_{t}$, and State-Node table $A$.}
\ENSURE {Multicast trees $T_{1}, T_{2}, \cdots, T_{t}$, $s_{i}$ is the root of $T_{i}$, and $D_{i}$ is in $T_{i}$.}
\STATE{//Multi-Tree Routing Phase}
\FOR{$i\in \{1, 2, \dots, t\}$}
	\STATE{$T_{i} \leftarrow$ shortest path tree containing $D_{i}$ with root $s_{i}$}
\ENDFOR
\STATE{}
\FOR{overloaded node $u$}
	\STATE{Reroute an appropriate downstream node $v$ of $u$ to balance the distribution of branch nodes}
\ENDFOR
\STATE{}
\STATE{//State-Node Assignment Phase 1) Greedy Assigning Stage}
\STATE{$A\leftarrow \left[ \mathbf{0} \right]$}
\WHILE{there is $x\in N-A$ such that $A\cup \{x\}\in\mathcal{M}$}
	\STATE{$x_{\max}\leftarrow \arg \max\limits_{x\in N - A}\{Z(A\cup \{x\}): A\cup\{x\}\in\mathcal{M}\}$}
	\STATE{$A \leftarrow A\cup \{x_{\max}\}$}
\ENDWHILE
\STATE{}
\STATE{//State-Node Assignment Phase 2) Local Search Stage}
\FOR{overloaded node $u$}
	\STATE{Re-assign node state on $u$ to maximizing reduction}
\ENDFOR
\STATE{}
\FOR{overloaded node $u$}
	\WHILE{node $u$ is overloaded}
		\FOR{each node $z$ in $V$}
			\IF{node $z$ is not overloaded}
				\STATE{Reroute node the downstream node $v$ of $u$ to node $z$}
				\STATE{Break the for loop}
			\ENDIF
		\ENDFOR
	\ENDWHILE
	%\STATE{Reroute an appropriate downstream node $v$ of $u$ to balance the distribution of branch nodes}
\ENDFOR
\RETURN{$T_{1}, T_{2}, \cdots, T_{t}$ and $A$}	
\end{algorithmic}
\end{algorithm}

\section{Performance Evaluation}
\label{sec:evaluation} In this section, we first compare MTRSA and other approaches with real topologies. Afterward, we deploy our algorithm in a small experimental SDN network with HP SDN switches to evaluate the video performance with YouTube HD traffic that requires a large amount of bandwidth consumption.. 

%The source code of RAERA and all our
%implementations (i.e., loss recovery in recovery nodes and a YouTube proxy
%for multicast) can be downloaded in ~\cite{extension_code}.

\subsection{Simulation Setup}
We simulate our algorithm in two real networks: VtlWavenet2011 and 
Columbus~\cite{_internet_2014}. VtlWavenet2011 includes 91 nodes and 
96 links, while Columbus has 70 nodes and 85 links. Our simulation is
divided into small-scale and large-scale cases. The numbers of trees in the small-scale cases are smaller than 100, whereas there are more than 2000 trees in the large-scale cases. The link capacity in the topologies is set to the level that the maximal bottleneck link utilization reaches 100\%~\cite{5743044}, and the edge cost of each link is set as 1.
We vary the number of multicast trees, the number of destinations, and node capacity. 
The source and destinations are chosen randomly from each network. 

We compare MTRSA with the following algorithms: 1) the shortest-path tree
algorithm (SPT), 2) the Steiner tree (ST) algorithm\footnote{There are some single-tree multicast routing algorithms with differernt purpuses (such as QoS), but they are not included in this study because ST (i.e., the optimal solution for single tree) outperforms those approaches in terms of the bandwidth consumption.}~\cite%
{hwang_steiner_1992}, and 3) CPLEX~\cite{_ibm_2014}, which finds the optimal solutions
of SMTE problem by solving the IP formulation in Section \ref%
{sec: ProblemFormulation}. In SPT and ST, the branch state nodes of different trees are randomly assigned to a branch node when node is fully utilized, i.e., the number of branch state nodes reaches the node capacity. Each SPT and ST is added to the  network iteratively. If an edge does not have sufficient residual capacity to support the multicast flow of a new SPT or ST, it will be removed accordingly to avoid choosing the edge in the SPT or ST. We implement all algorithms in an HP DL580 server with four Intel Xeon E7-4870 2.4 GHz CPUs and 128 GB RAM. Each simulation result is averaged over 100 samples.

\subsection{Small-Scale Evaluation}
In the small-scale cases, we compare the total bandwidth costs of all trees in MTRSA, SPT, ST, and CPLEX with different number of trees ($|T|$), different node
capacity ($b_u$), and different number of destinations $|D|$. As shown in Fig.~\ref{f:cost_cplex} and Fig.~\ref{f:cost_cplex_dst},
MTRSA generates a solution with the total bandwidth cost very close to the optimal
solution.  Although SPT chooses the shortest path to the  destinations, it does not carefully 
examine the node capacity, and its cost is thus higher than MTRSA.
Compared with SPT, the distance from the source to a destination in ST is 
usually higher because the path needs to be deviated from the shortest one in order to 
share more edges with another path. Therefore, more branch nodes are usually involved in ST. Without a sophisticated allocation of branch state nodes, ST incurs a slightly higher cost than SPT due to the additional bandwidth consumption in unicast tunneling for the branch nodes with full Group Table. The difference becomes more significant when the number of trees increases as shown in Fig.~\ref{f:cost_cplex}. Similarly, Fig.~\ref{f:cost_cplex_dst} manifests that the total bandwidth cost of each tree increases as the number of destinations grows, because each tree becomes larger in this case.

\begin{figure}[t!]
\centering
\subfigure[VtlWavenet2011]{
    \label{f:cost_cplex:a}
    \includegraphics[width=4cm, height = 2.2cm]{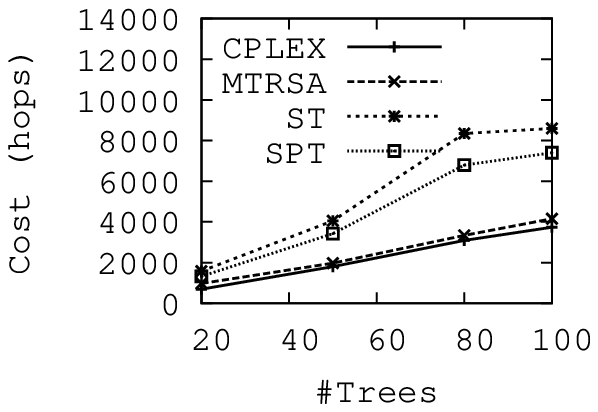}}
    %\vspace{-1.2\baselineskip}
\subfigure[Columbus]{
    \label{f:cost_cplex:b}
    \includegraphics[width=4.1cm, height = 2.2cm]{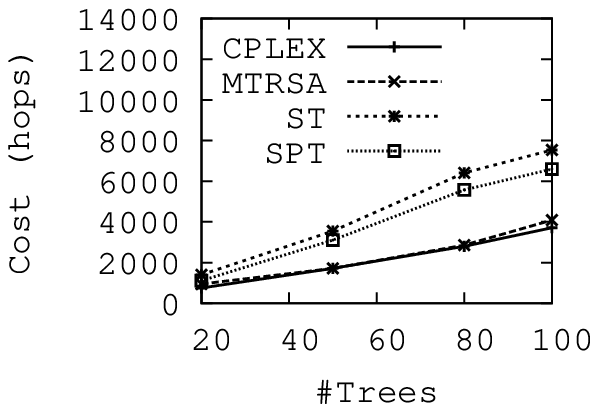}}
\caption{Cost with different $|T|$ ($b_u$ = 7, $|D|$ = 6)}
\label{f:cost_cplex}
\end{figure}
%\vspace{-10pt}

%\begin{figure}[t!]
%%\vspace{-10pt}
%\centering   
%\subfigure[VtlWavenet2011]{
%    \label{f:cost_cplex_cap:c}
%    \includegraphics[width=4cm, height = 2.2cm]{./figures/cost_cplex_cap_1}}
%    \vspace{-0.7\baselineskip}
%\subfigure[Columbus]{
%    \label{f:cost_cplex_cap:d}
%    \includegraphics[width=4cm, height = 2.2cm]{./figures/cost_cplex_cap_2}}
%\caption{Cost with different $b_u$ ($|T|$ = 100, $|D|$ = 6)}
%\label{f:cost_cplex_cap}
%\end{figure}

\begin{figure}[t!]
\centering  
%\vspace{-0.25in}
\subfigure[VtlWavenet2011]{
    \label{f:cost_cplex_dst:e}
    \includegraphics[width=4cm, height = 2.2cm]{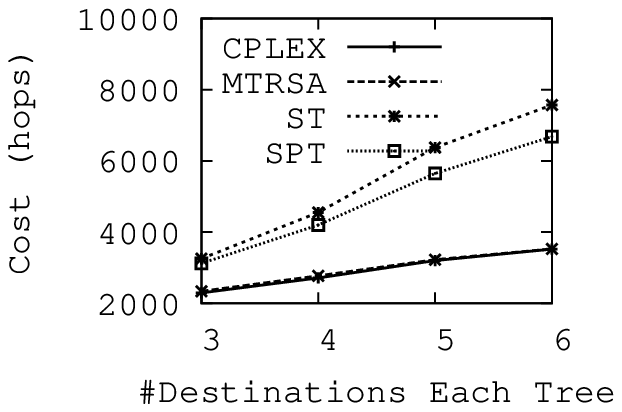}}
\subfigure[Columbus]{
    \label{f:cost_cplex_dst:f}
    \includegraphics[width=4cm, height = 2.2cm]{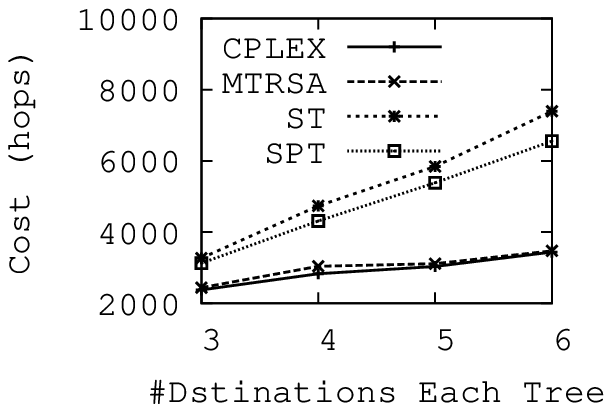}}
\caption{Cost with different $|D|$ ($|T|$ = 100, $b_u$ = 7)}
\label{f:cost_cplex_dst}
\end{figure}

%\begin{figure}[t!]
%\centering
%\includegraphics[width=8cm, height = 2.8cm]{cost_cplex.eps} \vspace{-0.05in} 
%%\vspace{-0.1in}
%%\vspace{-0.2in}
%\caption{Cost with different number of trees ($\bar{b_u}$ = 7, $\bar{|D|}$ = 6)}
%\label{f:cost_cplex}
%\end{figure}
%
%\begin{figure}[t!]
%\centering
%\includegraphics[width=8cm, height = 2.8cm]{cost_cplex_cap.eps} \vspace{-0.05in} 
%%\vspace{-0.1in}
%%\vspace{-0.2in}
%\caption{Cost with different node capacity ($|T|$ = 100, $\bar{|D|}$ = 6)}
%\label{f:cost_cplex_cap}
%\end{figure}
%
%\begin{figure}[t!]
%\centering
%\includegraphics[width=8cm, height = 2.8cm]{cost_cplex_dst.eps} \vspace{-0.05in} 
%%\vspace{-0.1in}
%%\vspace{-0.2in}
%\caption{Cost with different number of destinations ($|T|$ = 100, $\bar{b_u}$ = 7)}
%\label{f:cost_cplex_dst}
%\end{figure}

\subsection{Large-Scale Evaluation}
In the following, we evaluate MTRSA, ST, and SPT in larger-scale cases, where the number
of multicast tree is ranged from 2000 to 10000, the number of destinations is from
5 to 25, and the node capacity is between 50 and 250. Compared with smaller-scale 
cases, the advantage of MTRSA is more significant in larger-scale cases. In Fig.~\ref{f:cost_topo}, the total cost increases with the number of trees. For a larger network, the source and any destination are inclined to be located with distantly, but there is also a higher chance to find a node with sufficient capacity as a branch node for rerouting. Therefore, MTRSA effectively reduces the total bandwidth cost by 66\% and 59\%, respectively, compared to ST and SPT. In addition, Fig.~\ref{f:cost_topo_cap} shows that the bandwidth costs can be effectively reduced when we increase node capacity, and setting the node capacity as 100 is sufficient for MTRSA. On the other hand, the bandwidth cost grows with the number of destinations, because all trees are required to span more nodes as shown in Fig.~\ref{f:cost_topo_dst}.

%We also observe
%resource utilization, which is used node capacity divided by total node capacity,
%increases slowly with the number of destinations, because larger trees consumes
%more node capacity.

\begin{figure}[t!]
\centering
\subfigure[VtlWavenet2011]{
    \label{f:cost_topo:a}
    \includegraphics[width=4cm, height = 2.2cm]{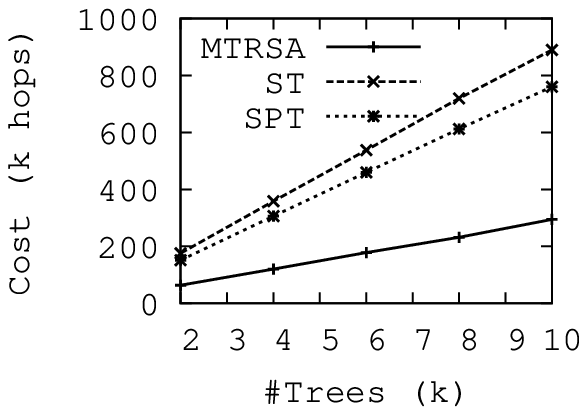}}
    %\vspace{-1.2\baselineskip}
\subfigure[Columbus]{
    \label{f:cost_topo:b}
    \includegraphics[width=4.1cm, height = 2.2cm]{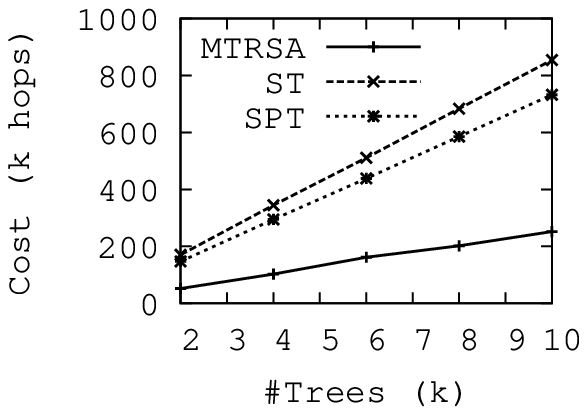}}
\caption{Cost with different $|T|$ ($b_u$ = 300, $|D|$ = 10)}
\label{f:cost_topo}
\vspace{-0.23in}
\end{figure}    

\begin{figure}[t!]
\centering 
\subfigure[VtlWavenet2011]{
    \label{f:cost_topo_cap:c}
    \includegraphics[width=4cm, height = 2.2cm]{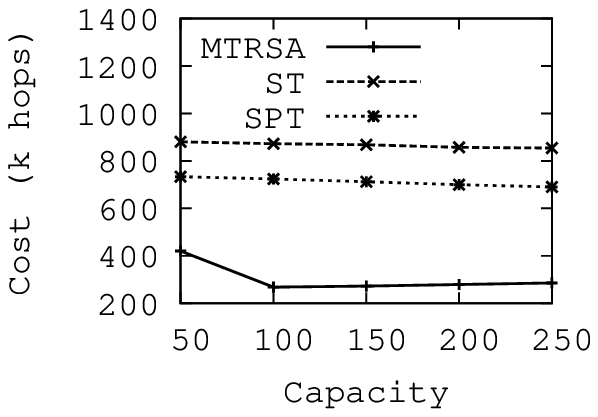}}
    \vspace{-0.7\baselineskip}
\subfigure[Columbus]{
    \label{f:cost_topo_cap:d}
    \includegraphics[width=4cm, height = 2.2cm]{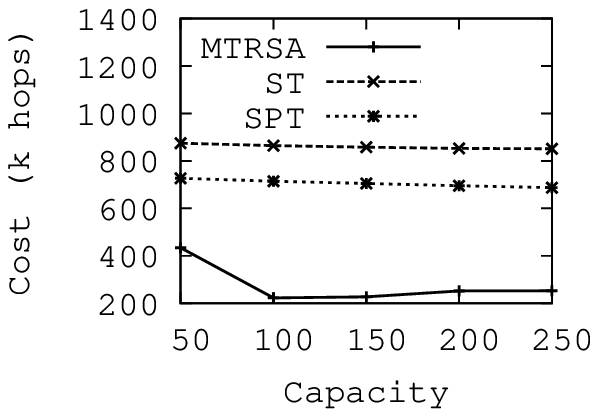}}
\caption{Cost with different $b_u$ ($|T|$ = 6000, $|D|$ = 10)}
\label{f:cost_topo_cap}
\vspace{-0.2in}
\end{figure}    

\begin{figure}[t!]
\centering    
%\vspace{-0.25in}
\subfigure[VtlWavenet2011]{
    \label{f:cost_topo_dst:e}
    \includegraphics[width=4cm, height = 2.2cm]{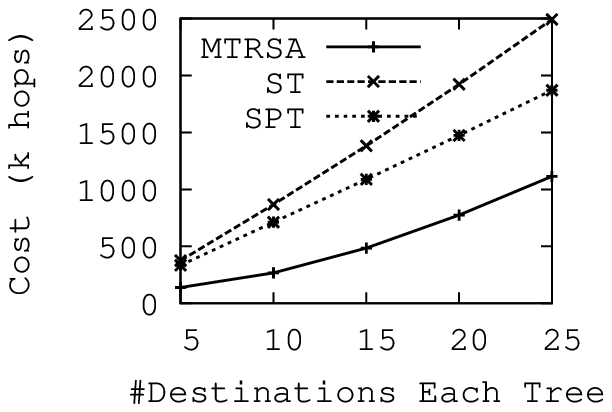}}
\subfigure[Columbus]{
    \label{f:cost_topo_dst:f}
    \includegraphics[width=4cm, height = 2.2cm]{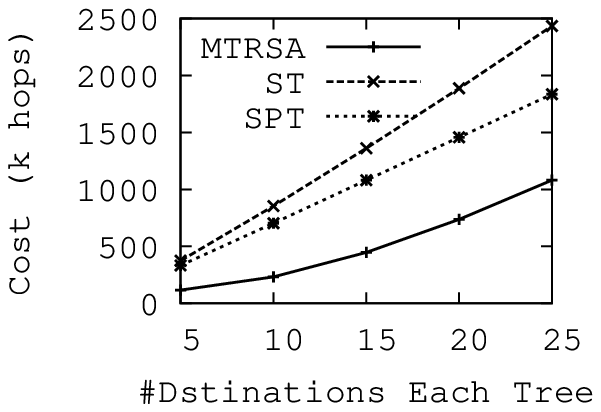}}
\caption{Cost with different $|D|$ ($|T|$ = 6000, $b_u$ = 300)}
\label{f:cost_topo_dst}
%\vspace{-0.23in}
\end{figure}

%\begin{figure}[t!]
%\centering    
%%\vspace{-0.25in}
%\subfigure[VtlWavenet2011]{
%    \label{f:utilization_topo_dst:e}
%    \includegraphics[width=4cm, height = 2.2cm]{./figures/utilization_topo_dst_1}}
%\subfigure[Columbus]{
%    \label{f:utilization_topo_dst:f}
%    \includegraphics[width=4cm, height = 2.2cm]{./figures/utilization_topo_dst_2}}
%\caption{Resource utilization with different $|D|$ ($|T|$ = 6000, $b_u$ = 300)}
%\label{f:utilization_topo_dst}
%\end{figure}

%\begin{figure}[t!]
%\centering
%\includegraphics[width=8cm, height = 2.8cm]{cost_topo.eps} \vspace{-0.05in} 
%%\vspace{-0.1in}
%%\vspace{-0.2in}
%\caption{Cost with different number of trees ($\bar{b_u}$ = 300, $\bar{|D|}$ = 10)}
%\label{f:cost_topo}
%\end{figure}
%
%\begin{figure}[t!]
%\centering
%\includegraphics[width=8cm, height = 2.8cm]{cost_topo_cap.eps} \vspace{-0.05in} 
%%\vspace{-0.1in}
%%\vspace{-0.2in}
%\caption{Cost with different node capacity ($|T|$ = 6000, $\bar{|D|}$ = 10)}
%\label{f:cost_topo_cap}
%\end{figure}
%
%\begin{figure}[t!]
%\centering
%\includegraphics[width=8cm, height = 2.8cm]{cost_topo_dst.eps} \vspace{-0.05in} 
%%\vspace{-0.1in}
%%\vspace{-0.2in}
%\caption{Cost with different number of destinations ($|T|$ = 6000, $\bar{b_u}$ = 300)}
%\label{f:cost_topo_dst}
%\end{figure}

%\begin{figure}[t!]
%\centering
%\includegraphics[width=8cm, height = 2.8cm]{utilization_topo_dst.eps} \vspace{%
%-0.05in} %\vspace{-0.1in}
%%\vspace{-0.2in}
%\caption{Resource utilization with different number of destinations ($|T|$ = 6000, $\bar{b_u}$ = 300)}
%\label{f:utilization_topo}
%\end{figure}

Table I summarizes the running time of MTRSA with different $|T|$ and
$|D|$. With a smaller input, such as 2000 trees and 5
destinations in each tree, the running time for MTRSA is around 1 second. As $|T|$ and 
$|D|$ increase, MTRSA only requires around 73
seconds in the largest case with 10000 multicast trees. Therefore, it is envisaged that our algorithm is practical to be deployed in SDN.

\begin{table}[t!]
\caption{Running time of MTRSA (seconds)}
\begin{center}
\begin{tabular}{|c|c|c|c|c|c|}
\hline
$|T|$ & $|D|$ = 5 & $|D|$ = 10 & $|D|$ = 15 & $|D|$ = 20 & $|D|$ = 25 \\ \hline
2000 & 1.00 & 2.29 & 3.01 & 3.43 & 3.99 \\ 
4000 & 3.09 & 6.33 & 8.71 & 11.06 & 12.39 \\ 
6000 & 5.82 & 12.52 & 17.94 & 22.37 & 25.53 \\ 
8000 & 8.80 & 20.10 & 29.89 & 37.52 & 43.78 \\
10000 & 12.83 & 32.32 & 48.44 & 57.01 & 73.40 \\
\hline
\end{tabular}
\label{t:running_time_k}
\end{center}
\par
%\vspace{-0.25in}
\vspace{-0.2in}
%\caption{The running time of RAERA in seconds}
\end{table}

\subsection{Implementation}
To evaluate MTRSA in real environments, we implement it in an experimental SDN
with HP Procurve 5406zl OpenFlow-enabled switches. We use Floodlight as the
OpenFlow controller to install the multicast forwarding rules in SDN-FEs. We install multicast group information in group table and create virtual ports mapping to multiple physical ports to forward multicast traffic.  MTRSA is
running on the top of Floodlight. Our testbed includes 12 nodes and 24 links as shown in Fig.~\ref{f:sdn_topo}, where the link capacity and node capacity are set as 50Mbps and 5, respectively. We randomly select 10 nodes as the video multicast sources, where each source is connected to a Youtube proxy to facilitate YouTube multicast. We implement the Youtube proxy by using VLC player, which can request video stream from Youtube and work as a video server. To support multicast, we modify TCP to aggregate TCP ACKs from multiple clients.
The full-HD test video is in 460 seconds encoded in H.264 with the average bit rate as 10Mbps. For each source, we randomly assign 10 destinations that play videos using the VLC player.
%Due to the space constraint, the detailed implementation setting is presented in \cite{extension}. 
Fig.~\ref{f:video_test} shows that the total bandwidth consumption during playback, and we average the bandwidth consumption every 40 seconds. The results manifest that the bandwidth consumption of MTRSA is 46\% and 35\% lower than ST and SPT, respectively. Therefore, MTRSA can effectively support multicast traffic engineering in SDN.

\begin{figure}[t!]
\centering
\subfigure[Topology]{
    \label{f:sdn_topo}
    \includegraphics[width=4cm, height = 3cm]{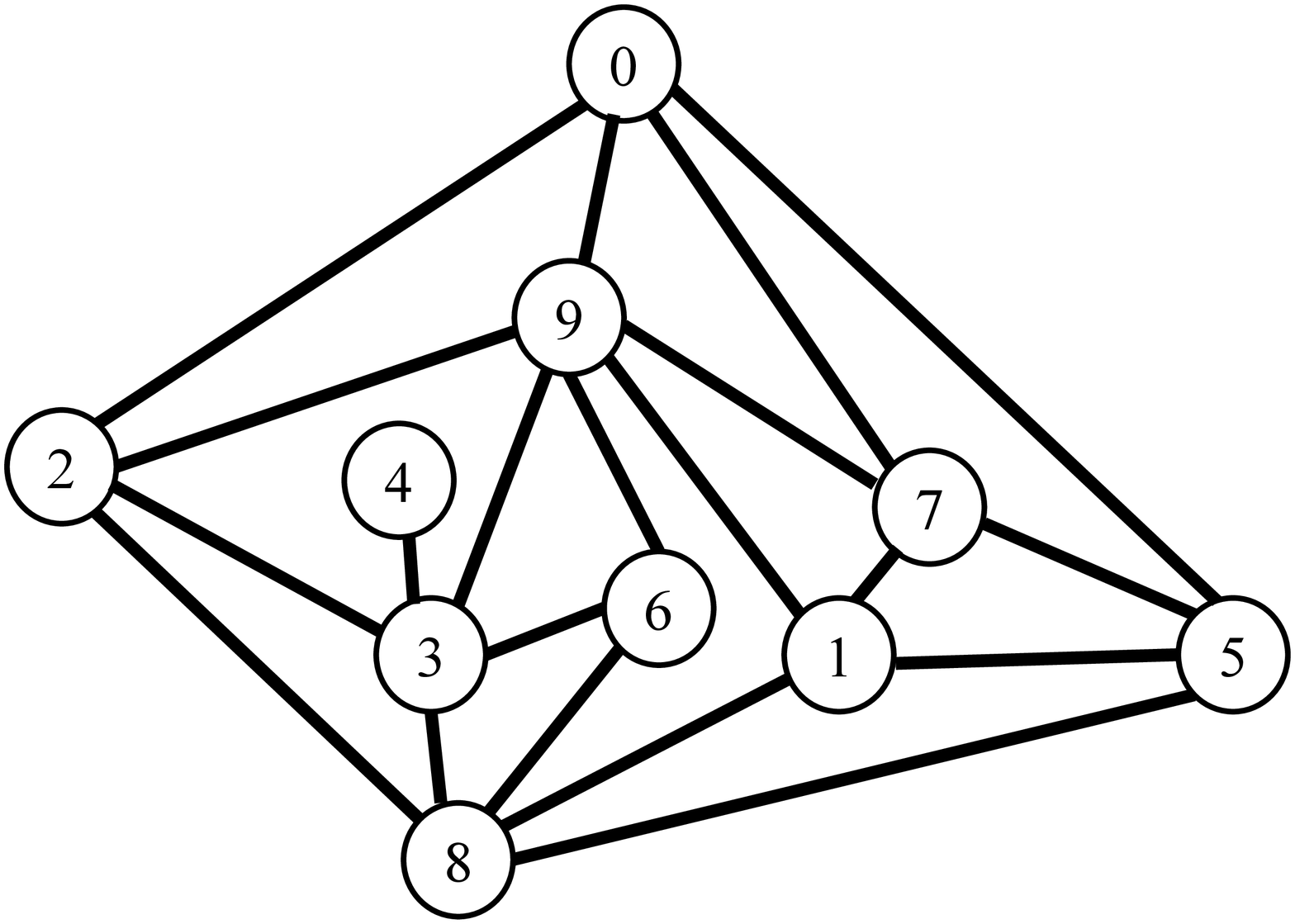}}
\subfigure[Bandwidth consumption]{
    \label{f:video_test}
    \includegraphics[width=4cm, height = 2.2cm]{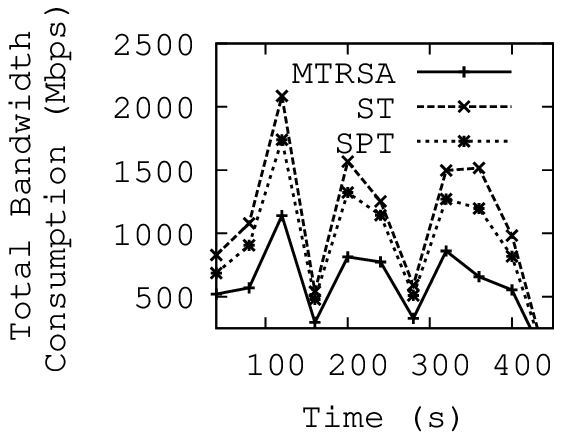}}
\caption{Implementation results of the experimental SDN}
\label{f:sdn_topo}
\vspace{-0.1in}
\end{figure}

%\begin{figure}[t!]
%\centering
%\includegraphics[width=4cm, height = 3cm]{video_test_topo.eps} \vspace{%
%-0.05in} %\vspace{-0.1in}
%%\vspace{-0.2in}
%\caption{The topology for video traffic experiment.}
%\label{f:sdn_topo}
%\end{figure}
%
%\begin{figure}[t!]
%\centering
%\includegraphics[width=6cm, height = 4cm]{video_test.eps} \vspace{%
%-0.05in} %\vspace{-0.1in}
%%\vspace{-0.2in}
%\caption{The topology for the experimental SDN network.}
%\label{f:video_test}
%\end{figure}

\section{Conclusion}
Recent studies on traffic engineering for SDN mostly focus on unicast, while most existing multicast routing algorithms are designed to find the routing of a multicast tree, instead of multiple trees. In this paper, therefore, we have formulated Scalable Multicast Traffic Engineering Problem (SMTE) to minimize the total bandwidth cost according to the link and node capacity constraints for multiple trees in SDN. We have proved that SMTE-N is NP-hard and not able to be approximated within $\delta $, while SMTE cannot be approximated within any ratio. To solve the problem, we have proposed Multi-Tree Routing and State Assignment Algorithm (MTRSA), which is a $\delta $-approximation algorithm for SMTE-N, while MTRSA has been extended to support SMTE as well. Simulation based on real topologies and implementation with Youtube traffic manifest that MTRSA can effectively find the routing of multiple multicast trees and assign the branch state nodes in order to reduce the total bandwidth cost, while the computation time to construct numerous trees is also reasonable for practical SDN. Since the tree obtained by MTRSA is not delay bounded, we will extend it to support QoS multicast in the future work.

\section*{Acknowledgment}
This work is supported by MOST 104-2622-8-009-001.

\linespread{0.98} 
\bibliographystyle{IEEEtran}
\bibliography{IEEEabrv,multicast}

\end{document}